\documentclass[journal,draftcls,onecolumn,12pt]{IEEEtran}
\usepackage[T1]{fontenc}

\usepackage{latexsym, graphicx, textcomp}
\usepackage{epsfig,upgreek}
\usepackage{amsfonts, amssymb, amsthm,makecell,enumerate}
\usepackage{float,color,dblfloatfix}
\usepackage{fancyhdr}
\usepackage{amsmath}
\usepackage{amsfonts}
\usepackage{mathrsfs}
\usepackage{pifont}
\usepackage{amssymb}
\usepackage{graphicx}
\usepackage{multicol}
\usepackage{cases}
\usepackage{epstopdf}
\usepackage{ltxtable, filecontents}
\usepackage[ruled,vlined,linesnumbered]{algorithm2e}

\usepackage{extarrows}
\usepackage{latexsym}
\usepackage{booktabs}
\usepackage{multirow}
\usepackage{indentfirst}
\usepackage{booktabs}
\usepackage{threeparttable}
\usepackage{supertabular}
\usepackage[figuresright]{rotating}
\usepackage[tight,small]{subfigure}
\usepackage{tabularx}
\usepackage{cite}
\usepackage{dcolumn}
\usepackage{hyperref}
\usepackage{bm}
\usepackage{url}
 \usepackage{algorithmic}
\usepackage{setspace}

\theoremstyle{definition}
\newtheorem{defn}{Definition}
\newtheorem{prop}{Proposition}

\graphicspath{{figures/}}
\allowdisplaybreaks

\begin{document}
%

\title{Let’s Trade in The Future! A Futures-Enabled Fast Resource Trading Mechanism in Edge Computing-Assisted UAV Networks}
%
%
%

\author{Minghui Liwang, \IEEEmembership{Member, IEEE}, Zhibin Gao, \IEEEmembership{Member, IEEE}, and Xianbin Wang, \IEEEmembership{Fellow, IEEE} 

%


}



%



\maketitle
\pagestyle{empty}  
\thispagestyle{empty} 

\begin{abstract}
Mobile edge computing (MEC) has emerged as one of the key technical aspects of the fifth-generation (5G) networks. The integration of MEC with resource-constrained unmanned aerial vehicles (UAVs) could enable flexible resource provisioning for supporting dynamic and computation-intensive UAV applications. Existing resource trading could facilitate this paradigm with proper incentives, which, however, may often incur unexpected negotiation latency and energy consumption, trading failures and unfair pricing, due to the unpredictable nature of the resource trading process. Motivated by these challenges, an efficient futures-based resource trading mechanism for edge computing-assisted UAV network is proposed, where a mutually beneficial and risk-tolerable forward contract is devised to promote resource trading between an MEC server (seller) and a UAV (buyer). Two key problems i.e. futures contract design before trading and power optimization during trading are studied. By analyzing historical statistics associated with future resource supply, demand, and air-to-ground communication quality, the contract design is formulated as a multi-objective optimization problem aiming to maximize both the seller’s and the buyer’s expected utilities, while estimating their acceptable risk tolerance. Accordingly, we propose an efficient bilateral negotiation scheme to help players reach a trading consensus on the amount of resources and the relevant price. For the power optimization problem, we develop a practical algorithm that enables the buyer to determine its optimal transmission power via convex optimization techniques. Comprehensive simulations demonstrate that the proposed mechanism offers both players considerable utilities, while outperforming the onsite trading mechanism on trading failures and fairness, negotiation latency, and cost.
\end{abstract}

\section{Introduction}
\noindent
\IEEEPARstart{D}{ue} to their deployment flexibility and cost-effectiveness, unmanned aerial vehicles (UAVs) have been widely deployed to support both public and civilian missions~\cite{1,2}, creating an exponential growing market exceeding \$125 billion worldwide~\cite{3}. The evolution of wireless technologies and the recent advances in the fifth generation (5G) networks have inspired a wide range of UAV applications, such as transportation management, disaster relief and rescue, UAV mounted base station, and smart surveillance~\cite{3,4}. Additionally, UAVs have realized significant values in recent blockbuster incidents. For example, UAVs enabled timely search around the helicopter crash involving basketball legend Kobe Bryant in 2020~\cite{5}. In an effort to slow the spread of COVID-19, UAVs have been employed to monitor public spaces and enforce social distancing rules all over the world~\cite{6}. 

However, many of the abovementioned UAV applications are generally computation-intensive and require complicated onboard processing algorithms and calculations, which pose significant challenges to UAVs with limited computing resources and capabilities. Additionally, the limited power supply of UAV poses another major difficulty for real-time data processing, networking, and decision-making, where performing intensive computations on a single UAV can bring adverse effects on both application completion and UAV's battery lifetime~\cite{7}. To address these issues, mobile edge computing (MEC), which emerges as one of the key technologies of 5G, has become a viable solution for UAV-enabled applications~\cite{8}, which brings the cloud capacity to the edge network, thus offers flexible and cost-effective computing services for resource constrained UAVs.

MEC-enabled resource provisioning to UAVs often relies on certain form of resource trading, where a UAV as the resource requestor with a heavy computational workload can offload a certain amount of tasks (applications) to the MEC server as the resource owner via a feasible access point (AP) or base station (BS) through air-to-ground (A2G) communication, while paying for the requested resources and computing service. To facilitate this resource provisioning paradigm with proper incentives, conventional onsite trading~\cite{9,10} has been widely adopted to enable the resource buying and selling among resource owners and requestors (collectively known as players). In particular, onsite players can reach a consensus on trading terms such as the resource amount and relevant price, based on the current resource supply, demand, and other network/platform/application-related conditions. Nevertheless, onsite trading may lead to undesirable performance degradation mainly incurred by the unpredictable nature of the trading process. First, the amount of time and cost (e.g., power and battery consumption) that onsite players have to devote for reaching a trading consensus are unpredictable, potentially resulting in unsatisfactory quality of service (QoS) and user experiences. Second, onsite players may face trading failures due to unpredictable factors such as varying wireless channel conditions. Additionally, onsite trading could be unfair, as evidenced by changing prices owing to the random nature of the trading process (e.g., unpredictable resource supply and wireless network condition). Major challenges when employing onsite resource trading mechanism in UAV networks are summarized as follows:

%
%
%

\noindent
$\bullet$ \textit{Latency of negotiation}: To reach a trading consensus, onsite negotiation between resource owner and requestor can result in excessive latency, which may dramatically reduce the net usable time for resource sharing under dynamic resource availabilities. Factors such as the delay-sensitive nature of real-time UAV applications, limited A2G connection duration, and varying network conditions (e.g., variable channel quality and changing topology of the UAV network) can substantially hinder the negotiation among players. Timely resource provisioning thus remains a major challenge under dynamic network conditions.

\noindent
$\bullet$ \textit{Energy consumption during negotiation}: An onsite player with limited battery lifetime and power supply may suffer from extra energy consumption to reach a trading consensus. The design of an energy-efficient trading mechanism is thus considered critical. 

\noindent
$\bullet$ \textit{Trading unfairness}: Unfairness is generally reflected by fluctuating prices (e.g., larger fluctuations exacerbate unfairness) due to random availability of resource supply and the greedy nature of resource owners. Therefore, mitigating undesirable trading unfairness under dynamic resource conditions represents another noteworthy challenge.

\noindent
$\bullet$ \textit{Trading failure}: Trading failures occur when players fail to reach a consensus (e.g., when trading leads to negative utility for either player), which inevitably prevent timely provisioning of expected resource sharing to UAVs and lead to an unsatisfactory user experience. Therefore, reducing potential trading failures is essential to support reliable and robust edge computing for future UAV applications.

To address the aforementioned challenges, we propose an efficient futures-based resource trading mechanism within an edge computing-assisted UAV (EC-UAV) network architecture. Specifically, \textit{futures} refers to a forward contract under which different players agree to buy or sell a certain number of commodities at a predetermined price in the future~\cite{10,11}. This arrangement greatly facilitates trading fairness (e.g., smooth pricing) and trading efficiency (e.g., low negotiation latency and energy/battery consumption) while reducing trading failures~\cite{12}. In this paper, we consider an MEC server (resource seller) and a UAV (resource buyer) as the two resource trading players. The proposed resource trading mechanism in this paper is formulated to solve the following key problems: 1) Forward resource trading contract design. This problem focuses on how to enable the seller and buyer to negotiate a consensus on contract term of the amount of trading resources and relevant unit price, by maximizing their expected utilities while evaluating the risk of possible losses; 2) Transmission power optimization. To achieve this goal, a practical algorithm is introduced that enables the buyer to determine its optimal transmission power during each trading. 
\vspace{-0.5cm}

\subsection{Related Work}

\noindent
$\bullet$ \textit{Task offloading and resource allocation in UAV networks}:
Several studies have focused on the task offloading and resource allocation problems in UAV networks, which can be roughly divided into two types based on the role that the UAV plays: 1) UAV serves as resource owner~\cite{13,14,15,16}, and 2) UAV acts as resource requestor~\cite{1,17,18,19}. In terms of UAVs possessing rather powerful capabilities,~\cite{13,14,15,16} studied the task offloading problem under a UAV-assisted edge computing network architecture where UAVs acted as computing servers that assisted on-ground mobile devices with computation-intensive tasks. With respect to UAVs as computing service requestors, Bai \textit{et al}.~\cite{1} devised an energy-efficient computation offloading technique for UAV-MEC systems with an emphasis on physical-layer security via convex optimization techniques. In~\cite{17}, Messous \textit{et al}. investigated the computation offloading problem in an EC-UAV network by establishing a non-cooperative theoretical game involving multiple players and three pure strategies. Messous \textit{et al}.~\cite{18} developed a sequential game containing three player types (i.e., UAV, base station, and MEC server) and confirmed the existence of a Nash equilibrium. Liwang \textit{et al}.~\cite{19} introduced a vehicular cloud-assisted graph task allocation problem in software-defined air-ground integrated vehicular networks. Tasks carried by UAVs were modeled as graphs and mapped to on-ground vehicles; moreover, a novel decoupled approach was proposed to solve the task and power allocation problem.

\noindent
$\bullet$ \textit{Resource trading mechanism}:
Research devoted to resource trading falls into two categories: 1) onsite trading, where players reach an agreement on trading terms according to the current conditions (e.g., present resource demand, supply, and channel quality), such as online games~\cite{17,18,20,21} and auctions~\cite{22,23,24}; and 2) futures-based trading, where players determine a forward contract over buying or selling a certain amount of resources at a reasonable price in the future~\cite{25, 26, 27, 28, 29, 30, 31, 32, 33, 34}. In~\cite{20}, Liwang \textit{et al}. studied the opportunistic computation offloading problem among moving vehicles by modeling two Stackelberg games under complete and incomplete information environments. A multi-user non-cooperative offloading game was investigated by Wang \textit{et al}.~\cite{21}, intended to maximize the utility of each vehicle via a distributed best response algorithm. Consider auction-based onsite trading, the VM allocation among edge clouds and mobile users was formulated as an $n$-to-one weighted bipartite graph matching problem by Gao \textit{et al}.~\cite{22}, based on a greedy approximation algorithm. In~\cite{23}, Liwang \textit{et al}. studied a Vickrey--Clarke--Groves-based reverse auction mechanism involving resource trading among vehicles and suggested a unilateral matching-based mechanism. Gao \textit{et al}.~\cite{24} developed a truthful auction for the computing resource trading market considering graph tasks. However, the random nature of networks usually brings onsite players undesirable trading failures and unfairness, which may lead to unsatisfactory QoS and user experiences. Additionally, the procedure to reach a trade-related consensus may result in heavy latency and energy consumption, which further complicate the trading mechanism design. In general, few studies have investigated the cost that players have to pay prior to each trading.

As a result, futures-based trading is emerged as a practical paradigm and has been extensively~adopted in financial and commodity exchange markets~\cite{10,11}, which greatly diminishes the number of trading failures, price fluctuations, and the negotiation latency and energy consumption before each trading. Studies associated with futures-based resource trading have mainly investigated the electricity market~\cite{25, 26, 27, 28, 29, 30}, spectrum trading~\cite{31, 32, 33}, and grid computing~\cite{34}. Khatib \textit{et al}.~\cite{25} proposed a systematic negotiation scheme through which a generator and load reached a mutually beneficial and risk-tolerable forward bilateral contract in mixed pool/bilateral electricity markets. Conejo \textit{et al}.~\cite{26} addressed a power producer's optimal involvement in a futures electricity market to hedge against the risk of pool price volatility. In~\cite{27}, Morales \textit{et al}. investigated scenario reduction techniques to accurately convey the uncertainties in futures market trading in electricity markets. Wang \textit{et al}.~\cite{28} introduced optimal dynamic hedging of electricity futures using copula-garch models. In~\cite{29}, Algarvio~\textit{et al}. addressed the challenge of maximizing power producers' profits (or returns) by optimizing their customer share in pool and forward markets. Mosquera-L\'{o}pez \textit{et al}.~\cite{30} assessed the drivers of electricity prices in spot and futures markets by considering German electricity markets.

The dynamic nature of wireless communication and scarcity of spectrum resources have necessitated the application of futures-enabled radio spectrum trading. Li \textit{et al}.~\cite{31} investigated a futures market in an effort to manage financial risk and discover future spectrum pricing. In~\cite{32}, Gao \textit{et al}. studied a hybrid secondary spectrum market involving the futures and spot market, where buyers could purchase part of the underutilized licensed spectrum from a spectrum regulator through either predefined contracts or spot transactions. Sheng \textit{et al}.~\cite{33} proposed a futures-based spectrum trading mechanism to alleviate the risk of trading failures and unfairness. When evaluating resource trading in grid computing, Vanmechelen \textit{et al}.~\cite{34} proposed a hybrid market combining the spot and futures markets to achieve economical grid resource management. 

Nevertheless, futures-based trading has rarely been investigated in edge computing-assisted networks. To the best of our knowledge, this work is among the first to consider the design of a trading mechanism intended for future computing resources in EC-UAV networks. 

\vspace{-0.5cm}
\subsection{Novelty and Contribution}
\noindent
In this paper, futures is applied to facilitate the fast, and mutually beneficial resource trading between an MEC server (seller) and a UAV (buyer), under the EC-UAV network architecture. Our major contributions are summarized as follows:

\noindent
$\bullet$ To relieve the unexpected negotiation latency, cost, trading failures, and unfair prices, a novel futures-based resource trading mechanism under EC-UAV network architecture is proposed, which enables the players to negotiate a mutually beneficial and risk-tolerable forward contract of the amount of trading resources and the relevant price, which will be fulfilled during each trading in the future.

\noindent
$\bullet$ To tackle the unpredictable nature of the resource trading system, three main uncertainties are considered: the varying number of local users of the seller (resource supply), the buyer's task arrival rate (resource demand), and ever-changing A2G channel qualities. We correspondingly formulate the seller's utility as the weighted sum of local revenue, trading income, and possible refunds incurred by selling superabundant resources. Additionally, we define the buyer's utility as the trade-off among the saved task completion time from enjoying the computing service, the payment for the requested resources, and the energy consumption when offloading certain tasks to the seller via A2G channel. More importantly, we evaluate the risk facing both players to mitigate potential losses by analyzing historical uncertainty-related statistics.

\noindent
$\bullet$ The proposed resource trading mechanism is considered upon addressing two key problems: the forward contract design problem (before trading) and the buyer's transmission power optimization problem (during each trading). The former is formulated as a multi-objective optimization (MOO) problem where the seller and buyer each aims to maximize the expected utility, by analyzing the historical statistics of resource supply, demand, and A2G channel quality, while estimating the tolerable risk. To tackle this problem, we propose an efficient bilateral negotiation scheme that facilitates the players reaching a trading consensus. For the latter problem which involves maximizing the buyer's utility when fulfilling the contract during each future trading, we present a practical optimization algorithm which enables the buyer to obtain the optimal transmission power via convex optimization techniques.

\noindent
$\bullet$ Comprehensive simulation results and performance evaluation demonstrate that the proposed futures-based resource trading mechanism offers mutually beneficial utilities to both players, while outperforming the onsite trading mechanism on critical indicators including trading fairness and failures, negotiation latency and cost.

The rest of this paper is organized as follows. In Section II,
we introduce the system model and the problem overview. The problem formulation and the relevant solution of the proposed futures-based resource trading mechanism are detailed in Section III and Section IV, respectively. Numerical results are analyzed in Section V before drawing the conclusion in Section VI.

\section{System Model and Problem Overview}

\subsection{Proposed Framework and Problem Overview}
\noindent
The proposed futures-based resource trading mechanism in EC-UAV networks primarily consists of an MEC server (seller) as the resource owner and a UAV (buyer) as the resource requestor. The seller owns a collection of computing resources that can run a certain amount of tasks in parallel. For analytical simplicity, we use virtual machine (VM)-based representation to quantize available resources, denoted by $V$ (e.g., $V=12 $ in Fig.~1). A local user of the seller is seen as a user who paid for the membership of the computing service, while each local user has a computation-intensive task to be processed during a trading. Notably, the number of local users is denoted by $n_l$, which follows a discrete uniform distribution\footnote{The discrete uniform distribution is denoted by ${\rm U}^{\rm d}()$ in this paper.} $n_l\sim {\rm U}^{\rm d}(0, 1, 2, \cdots, M)$, and $M>V$. The buyer features limited computing capability, resources, and power supply, which requires to process a number of tasks in each trading. Consequently, the buyer can generally offload a certain number of tasks to the seller for execution by paying for the relevant computing service. Suppose that the buyer's task\footnote{In this paper, we assume that both local users and the buyer's tasks randomly arrive~\cite{35,36,37} in the system and obey a discrete uniform distribution to convey uncertainties in the resource trading environment. Specifically, $n_b\neq 0$ indicates that the buyer always has task(s) during each trading, and $N$ denotes the relevant maximum number of tasks.} arrives at rate $n_b $ in each trading, which obeys the discrete uniform distribution $n_b\sim {\rm U}^{\rm d}(1, 2, \cdots, N)$, where $N>V$. 

This paper investigates a novel futures-based resource trading mechanism where the players determine the amount of trading resources $\mathcal{A}$ and the relevant unit price $\mathcal{P}$ in~advance~via a mutually beneficial forward contract, relying on historical statistics\footnote{The statistics are assumed to be known based on the historical records of each player~\cite{33}. However, information privacy exists in this system; for instance, the seller is unaware of the distributions of $ n_b$ and $\gamma$ , while the buyer is unaware of~information related to $n_l$, $p_l$, and $r_l$.} associated with resource supply (e.g., $n_l $), resource demand (e.g., $n_b $), and network conditions (e.g., A2G channel quality), and each future trading can be fulfilled accordingly. Moreover, in an effort to ensure greater utility for the buyer during each trading, this paper also studies a power optimization problem to facilitate an energy-efficient resource trading system. Fig. 1 illustrates the relevant framework, along with several examples of the proposed futures-based resource trading mechanism. Specifically, the timeline is divided into two segments: before trading and future trading, where the players negotiate a forward contract in advance before trading; based on which, each future trading can be implemented efficiently. 

\begin{figure}[h!t]
\centerline{\includegraphics[width=1\linewidth]{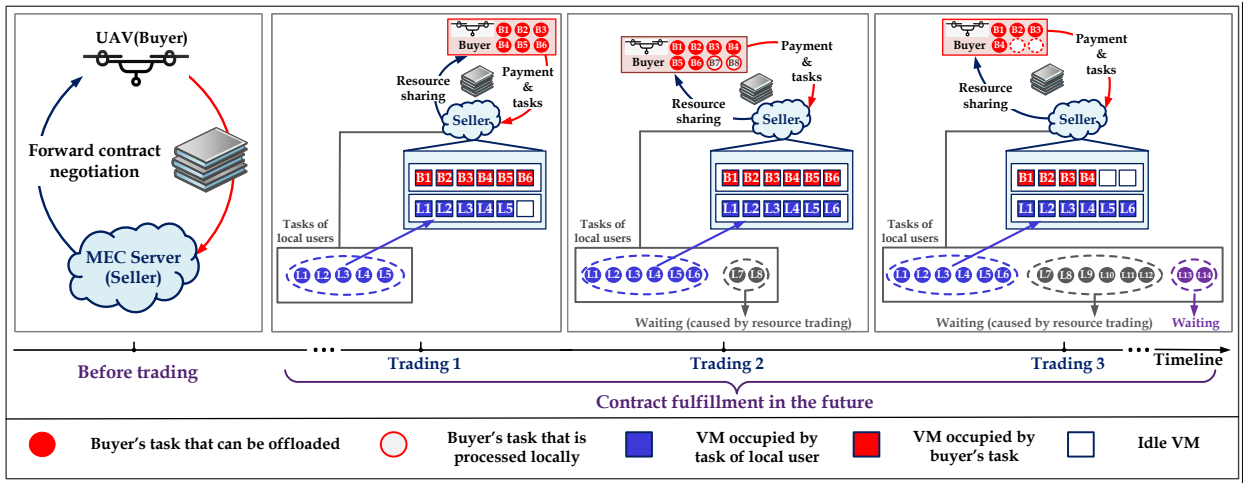}}
\vspace*{-.6em}
\caption{The proposed futures-based resource trading framework, and relevant examples ($V=12$, $\mathcal{A}^{*}=6$). The two players negotiate a forward contract before trading, which will be fulfilled in the future. In trading 1, $n_l=5$, $n_b=6$, ${[\mathcal{A}^{*}, n_b]}^-=6$, $C^s=0$; in trading 2, $n_l=8$, $n_b=8$, ${[\mathcal{A}^{*}, n_b]}^-=6$, $C^s=2r_l $; in trading 3, $n_l=14$, $n_b=4$, ${[\mathcal{A}^{*}, n_b]}^-=4$, $C^s=6r_l $.}
\label{fig1}
\end{figure}


\vspace{-0.5cm}
\subsection{Modeling of the Seller}
\noindent
Suppose that the seller can only accept a trade within its capability ($\mathcal{A}\le V$). In this section, the seller's utility, expected utility, and risk are described in detail.

\noindent
$\bullet$ \textit{Seller's utility}: The seller's utility consists of three components: 1) the local revenue $U^s=n_l\times p_l$, where $n_l $ and $p_l $ denote the number and unit revenue of local users, respectively; 2) the income $\mathcal{A}\times \mathcal{P}$ obtained from trading; and 3) the cost $C^s$ incurred by resource trading. Specifically, the cost $C^s$ is considered as the total refund for local users who have to wait for the release of occupied VMs caused by trading\footnote{In this paper, we only consider the waiting cost (refund) incurred by resource trading for the seller's utility. For example, when $n_l>V$, $n_l-V+\mathcal{A}$ tasks have to wait for the release of VMs, but only $\mathcal{A}$ of them are caused by resource trading with the buyer. For example, the seller pays no refund to the tasks in purple in trading 3, Fig.~1.} (e.g., tasks in grey in trading 2 and trading 3 of Fig.~1), defined as (1). 
\begin{align}
\label{eq1}
C^s=
\begin{cases}
0, & 0\le n_l\le V-\mathcal{A} \vspace{-2ex} \\ 
r_l(n_l-(V-\mathcal{A})), & V-\mathcal{A}<n_l\le V \vspace{-2ex} \\
r_l\mathcal{A}, & V+1\le n_l\le M 
\end{cases}, 
\end{align}
where $r_l $ ($r_l\le p_l $) denotes the refund to a local user when the remaining VMs after resource trading fail to meet current local task requirements. Consequently, the seller's utility is given as: 
\begin{align}
\label{eq2}
\mathcal{U}^s(n_l, \mathcal{A}, \mathcal{P})=U^s+\mathcal{A}\times \mathcal{P}-C^s.
\end{align}

\noindent
$\bullet$ \textit{Seller's expected utility}: The unpredictable number of local users poses challenges to the seller in terms of guaranteeing utility when fulfilling the forward contract in the future. Thus, we calculate the seller's expected utility $\overline{\mathcal{U}^s}(n_l, \mathcal{A}, \mathcal{P})$ as (3), where the relevant derivation is given by the Appendix. 
\begin{align}
\label{eq3}
& \overline{\mathcal{U}^s}(n_l, \mathcal{A}, \mathcal{P})=\frac{-r_l\mathcal{A}^2}{2(M+1)}+\left(\mathcal{P}-\frac{(r_l+2Mr_l-2Vr_l)}{2(M+1)}\right)\mathcal{A}+\frac{p_lM}{2}
\end{align}


\noindent
$\bullet$ \textit{Seller's risk}:
In the proposed resource trading system, risk is largely derived from the randomness of resource supply, resource demand, and A2G channel quality. To relieve the possible risks of loss during each future trading, we define the seller's risk as the probability that its utility $\mathcal{U}^s(n_l, \mathcal{A}, \mathcal{P})$ is smaller than or close to the expectation $\overline{\mathcal{U}^s}(n_l, \mathcal{A}, \mathcal{P})$ as given in (4), where $\lambda^s_1 $ is a threshold approaching 1. Namely, the seller always prefers a higher utility that exceeds its expectation in each trading. 
\begin{align}
\label{eq4}
\mathcal{R}^s(n_l, \mathcal{A}, \mathcal{P})=\Pr \{\mathcal{U}^s(n_l, \mathcal{A}, \mathcal{P})\le \lambda^s_1\times \overline{\mathcal{U}^s}(n_l, \mathcal{A}, \mathcal{P})\} 
\end{align}
By combining (1), (2), and (3), we rewrite (4) as (5).
\begin{align}
\label{eq5}
& \mathcal{R}^s(n_l, \mathcal{A}, \mathcal{P})=\Pr \bigg\{S\le \lambda^s_1\bigg(\frac{-r_l\mathcal{A}^2}{2(M+1)}+\bigg(\mathcal{P}-\frac{(r_l+2Mr_l-2Vr_l)}{2(M+1)}\bigg)\mathcal{A}+\frac{p_lM}{2}\bigg)-\mathcal{A}\mathcal{P}\bigg\},
\end{align}
where $S=U^s-C^s$ denotes a discrete random variable for notational simplicity, representing the left-hand side of ``$\le $'', and $\mathbb{C}_1=\lambda^s_1\times \left(-\frac{r_l}{2(M+1)}\mathcal{A}^2+\left(\mathcal{P}-\frac{(r_l+2Mr_l-2Vr_l)}{2(M+1)}\right)\mathcal{A}+\frac{p_lM}{2}\right)-\mathcal{A}\mathcal{P}$ represents the right-hand side of ``$\le $'' in (5). Correspondingly, $S$ is considered in (6).
\begin{align}
\label{eq6}
&S=
\begin{cases}
n_lp_l, & 0\le n_l\le V-\mathcal{A}  \vspace{-2ex} \\ 
n_lp_l-{n_lr}_l+r_l(V-\mathcal{A}), & V-\mathcal{A}<n_l\le V \vspace{-2ex} \\ 
n_lp_l-r_l\mathcal{A}, & V+1\le n_l\le M 
\end{cases}
\end{align}
Through analyzing the probability mass function (PMF) of $S$, the seller's risk can be recalculated by (7), where $\left\lfloor \cdot \right\rfloor $ denotes the rounded down operation (see detailed derivation in the Appendix).


 \vspace*{-0.8\baselineskip}  
\begin{align}
\label{eq7}
& \mathcal{R}^s(n_l, \mathcal{A}, \mathcal{P})
=
\begin{small}
\begin{cases}
0, & \mathbb{C}_1<0 \vspace{-0.6ex} \\ 
\dfrac{\left\lfloor \dfrac{\mathbb{C}_1}{p_l}\right\rfloor +1}{M+1}, & \makecell[l]{0\le \mathbb{C}_1<(V-\mathcal{A}) p_l +p_l-r_l} \vspace{-0.6ex} \\ 
\dfrac{V-\mathcal{A}+1}{M+1}+\dfrac{\left\lfloor \dfrac{\mathbb{C}_1-(V-\mathcal{A}) p_l}{p_l-r_l}\right\rfloor} {M+1}, & \makecell[l]{(V-\mathcal{A}) p_l+p_l-r_l\le \mathbb{C}_1  <(V+1) p_l-r_l\mathcal{A}} \vspace{-0.6ex} \\[3mm] 
\dfrac{V+1}{M+1}+\dfrac{\left\lfloor \dfrac{\mathbb{C}_1+r_l\mathcal{A}}{p_l}\right\rfloor} {M+1}, & \makecell[l]{(V+1) p_l-r_l\mathcal{A}\le \mathbb{C}_1\le Mp_l-r_l\mathcal{A}} \vspace{-0.6ex} \\ 
1, & \mathbb{C}_1>Mp_l-r_l\mathcal{A} 
\end{cases}
\end{small}
\end{align}

\vspace{-0.5cm}
\subsection{Modeling of the Buyer}
\noindent
In this paper, consider a buyer that faces difficulties in processing computation-intensive tasks locally due to insufficient resources, limited capability, and battery lifetime. 

\noindent
$\bullet$ \textit{Buyer's utility}:
The buyer's utility $\mathcal{U}^b$ is defined as the benefit obtained from resource trading, which mainly involves three features: 1) the task completion time $U^b$ saved from the computing service; 2) the relevant payment $\mathcal{A}\times \mathcal{P}$; and 3) the energy consumption $E^b$ incurred by offloading tasks from the buyer to the seller via A2G communication. Specifically, $U^b$ represents the difference between the completion time of a certain number of tasks by local computing (by the UAV itself) and edge cloud computing (by the seller), shown as (8):
\begin{align}
\label{eq8}
U^b={[\mathcal{A}, n_b]}^-\times \tau ^b-\left(\tau ^s+\frac{{[\mathcal{A}, n_b]}^-\times D}{W\log _2(1+q\gamma)}\right), 
\end{align}
where ${[\mathcal{A}, n_b]}^-$ refers to the smaller value between $\mathcal{A}$ and $n_b$, describing the actual number of tasks that can be offloaded to the seller (e.g., the actual number of offloaded tasks is calculated as ${[8, 6]}^-=6$ in trading 2, Fig.~1); and $\tau ^s$ and $\tau ^b$ denote the execution time of a task processed by the seller and the buyer, respectively\footnote{In the proposed EC-UAV network, the seller provides parallel processing service given multiple VMs, while the buyer works in a serial processing mode. For example, the execution time for 9~tasks is $\tau ^s$ by seller, while that of the buyer is $9\times \tau ^b$.}. The terms ${[\mathcal{A}, n_b]}^-\times \tau ^b$ and $\tau ^s+\frac{{[\mathcal{A}, n_b]}^-\times D}{W\log _2(1+q\gamma)}$ refer to the completion time of ${[\mathcal{A}, n_b]}^-$ tasks via local computing, and edge computing, respectively. Specifically, ${D}/{W\log _2(1+q\gamma)}$ in (8) represents the data transmission delay when offloading a task to the seller, where $D$ is the data size (bits) of each task\footnote{In this paper, assume that all the tasks have the same data size for analytical simplicity, and our proposed mechanism can be implemented effectively when considering different data sizes.}, $W$ indicates the bandwidth of the A2G channel, and $q~(0<q\le q^{max})$ stands for the buyer's transmission power. Moreover, $\gamma ={g_1d^{-\alpha}}/{N_0}$, where $g_1 $ corresponds to the channel gain at the reference distance of 1~meter, $d$ denotes the A2G distance between the buyer and the nearest AP, $\alpha $ indicates the path loss exponent of the line-of-sight path~\cite{38,39}, and $N_0 $ denotes the background noise power. 

In this paper, suppose that the UAV moves randomly in the sky within a certain space. Therefore, let $\gamma $ follow a uniform distribution in interval $[\varepsilon _1, \varepsilon _2]$ given the uncertainty of the wireless communication environment~\cite{33}, denoted by $\gamma \sim {\rm U}(\varepsilon _1, \varepsilon _2)$. Apparently, $q\gamma $ represents the seller's received signal-to-noise ratio (SNR). Hence, the energy consumption $E^b$ is considered as the buyer's extra overhead when transmitting a certain amount of data to the seller via wireless access as shown in (9):
 \vspace*{-0.4\baselineskip}  
\begin{align}
\label{eq9}
E^b=\frac{q\times {[\mathcal{A}, n_b]}^-\times D}{W\log _2(1+q\gamma)}+\ell, 
\end{align}
where $\ell $ indicates the tail energy, given that the UAV will hold the channel for a while even after data transmission. Correspondingly, the buyer's utility $\mathcal{U}^b$ is given in (10), where $\omega _1$ and $\omega _2$ are positive-weight~coefficients.
\vspace*{-0.4\baselineskip}  
\begin{align}
\label{eq10}
\mathcal{U}^b(q, \gamma, n_b, \mathcal{A}, \mathcal{P})=U^b-\omega _1\mathcal{A}\times \mathcal{P}-\omega _2E^b 
\end{align}

\noindent
$\bullet$ \textit{Buyer's expected utility}:
Similar to the seller, we compute the buyer's expected utility by (11) given the distributions of $\gamma $ and $n_b $ (find detailed derivation in the Appendix):
\begin{align}
\label{eq11}
& \overline{\mathcal{U}^b}(q, \gamma, n_b, \mathcal{A}, \mathcal{P})=\frac{-\mathbb{C}_2\mathcal{A}^2}{2N}+\left(\mathbb{C}_2+\frac{\mathbb{C}_2}{2N}-\omega _1\mathcal{P}\right)\mathcal{A}-\tau ^s-\omega _2\ell, 
\end{align}
where $\mathbb{C}_2=\tau ^b-\frac{D+\omega _2qD}{W}\times \frac{\ln 2\times \int^{y_2}_{y_1}{\frac{e^x}{x}}{\rm d}x}{q(\varepsilon _2-\varepsilon _1)}$ denotes a constant under any given $q$, $y_1=\ln 2\times \log _2(q\varepsilon _1+1)$, and $y_2=\ln 2\times \log _2(q\varepsilon _2+1)$ for notational simplicity. 

\noindent
$\bullet$ \textit{Buyer's risk}:
To alleviate heavy on-board workload, the buyer is consistently willing to trade with the seller when $\mathcal{U}^b>0$. Nevertheless, a trading may suffer from poor A2G channel quality due to particular factors (e.g., small $q$ and $\gamma $), which thus leads to unsatisfactory $\mathcal{U}^b$. Therefore, we define the minimum utility $\mathcal{U}^{min}$ of the buyer as a value approaching zero, describing a case in which all the tasks have to be processed locally (e.g., a failed trading). Consequently, the buyer's risk $\mathcal{R}^b(q, \gamma, n_b, \mathcal{A}, \mathcal{P})$ is largely tied to the prediction uncertainty of the randomness of resource demand $n_b $ and channel condition $\gamma$, which is formulated as the probability that $\mathcal{U}^b$ might be too close to its minimum $\mathcal{U}^{min}$ as shown in (12):
\begin{align}
\label{eq12}
\mathcal{R}^b(q, \gamma, n_b, \mathcal{A}, \mathcal{P})=\Pr \left\{\frac{\mathcal{U}^b(q, \gamma, n_b, \mathcal{A}, \mathcal{P})}{\mathcal{U}^{min}}\le \lambda^b_1+1\right\}, 
\end{align}
where $\lambda^b_1 $ represents a threshold coefficient. Upon integrating (8)-(11), (12) is rewritten as (13).
\vspace*{-0.4\baselineskip}  
\begin{align}
\label{eq13}
& \mathcal{R}^b(q, \gamma, n_b, \mathcal{A}, \mathcal{P})=
\Pr\bigg\{{[\mathcal{A}, n_b]}^-\left(\tau ^b-\frac{D+\omega _2qD}{W\log _2(1+q\gamma)}\right)\le (\lambda^b_1+1)\mathcal{U}^{min}+\tau ^s+\omega _1\mathcal{A}\mathcal{P}+\omega _2\ell \bigg\}
\end{align}
For notational simplicity, let $\mathbb{C}_{3}= [\mathcal{A}, n_b]^- $ $\left(\tau ^b-\frac{D+{\omega}_{2}qD}{W\log _{2}(1+q\gamma)}\right)$ and $\mathbb{C}'_{3}=(\lambda^b_{1}+1)\mathcal{U}^{min}+\tau ^s+{\omega}_{1}\mathcal{A}\mathcal{P}+$ ${\omega}_{2}\ell $ indicate the left-hand and right-hand side of ``$\le $'' in (13), respectively. Moreover, let discrete~random~variable $X={[\mathcal{A}, n_b]}^-$ and continuous~random~variable $Z=\tau^b-$ $\frac{D+{\omega}_{2}qD}{W\log _{2}(1+q\gamma)}$; the CDF ${\rm F}_Z(z)$ of $Z$ is calculated as (14):
\begin{align}
\label{eq14}
{\rm F}_Z(z)=
\begin{cases}
0, & z<\mathbb{C}_{4}\vspace{-1ex} \\ 
\dfrac{{2}^{\frac{D+{\omega}_{2}qD}{W(\tau ^b-z)}}-q{\varepsilon}_{1}-1}{{q\varepsilon}_{2}-{q\varepsilon}_{1}}, & \mathbb{C}_{4}\le z\le \mathbb{C}'_{4} \vspace{-1ex} \\ 
1, & z>\mathbb{C}'_{4} 
\end{cases}, 
\end{align}
where $\mathbb{C}_{4}=\tau ^b-\frac{D+{\omega}_{2}qD}{W\log _{2}({q\varepsilon}_{1}+1)}$ and $\mathbb{C}'_{4}=\tau ^b-\frac{D+{\omega}_{2}qD}{W\log _{2}({q\varepsilon}_{2}+1)}$ for notational simplicity. According to (14), we can compute $\mathcal{R}^b(q, \gamma, n_b, \mathcal{A}, \mathcal{P})$ when $\mathcal{A}=1$ as (15):
\begin{align}
\label{eq15}
& \mathcal{R}^b(q, \gamma, n_b, \mathcal{A}=1, \mathcal{P})={\rm F}_Z(\mathbb{C}'_{3})=
\begin{cases}
0, & \mathbb{C}'_{3}<\mathbb{C}_{4} \vspace{-0.6ex} \\ 
\dfrac{{2}^{\frac{D+{\omega}_{2}qD}{W(\tau ^b-\mathbb{C}'_{3})}}-q{\varepsilon}_{1}-1}{{q\varepsilon}_{2}-{q\varepsilon}_{1}}, & \mathbb{C}_{4}\le \mathbb{C}'_{3}\le \mathbb{C}'_{4} \vspace{-0.6ex} \\
1, & \mathbb{C}'_{3}>\mathbb{C}'_{4} 
\end{cases},
\end{align}
and (16) when $\mathcal{A}>1$, respectively. The derivations of (14)--(16) are detailed in the Appendix.
\label{eq16}
\begin{align}
& \mathcal{R}^b(q, \gamma, n_b, \mathcal{A}>1, \mathcal{P})=\Pr \{XZ\le \mathbb{C}'_{3}\}
=\frac{1}{N}\sum^{x=\mathcal{A}-1}_{x=1}{{\rm F}_Z\left(\frac{\mathbb{C}'_{3}}{x}\right)}+\frac{N-\mathcal{A}+1}{N}{\rm F}_Z\left(\frac{\mathbb{C}'_{3}}{\mathcal{A}}\right)
\end{align}
\subsection{Contract Term}
\noindent
Contract term represents a key concept in this paper, which determines the basis of all the future resource trading.

\begin{defn}[Final contract term]
\label{defn1}
The final contract term is denoted by $\{\mathcal{A}^{*}, \mathcal{P}^{*}\}$, referring to the buyer's and seller's final consensus on the forward contract. In a futures-based resource trading environment, each trading should be implemented based on the final contract term.
\end{defn}

\begin{defn}[Candidate contract term]
\label{defn2}
A candidate contract term refers to a pair of $\mathcal{A}$ and $\mathcal{P}$ that is accepted by the seller and buyer. The final contract term $\{\mathcal{A}^{*}, \mathcal{P}^{*}\}$ will be chosen from the set of candidate contract terms.
\end{defn}

\section{Problem Formulation}
\noindent
Note that contract period is beyond the scope of this paper\footnote{We do not consider the contract period in this paper; essentially, the two players can negotiate another forward contract when the previous one is about to expire.}. The proposed futures-based resource trading aims to solve two key problems: contract design and transmission power optimization. First, the seller and buyer determine a forward contract on resource amount and unit price by maximizing their expected utility, which is formulated as a MOO problem $\bm{\mathcal{F}_1}$ given in (17): 
\begin{align}
\centering
\label{eq17}
&\bm{\mathcal{F}_1}: 
\begin{cases}
\{\mathcal{A}^{*}, \mathcal{P}^{*}\}=\mathop{\arg\max}\limits_{\mathcal{A}, \mathcal{P}} {\overline{~\mathcal{U}^s}(n_l, \mathcal{A}, \mathcal{P})} \vspace{-1ex} \\[3mm]
\{\mathcal{A}^{*}, \mathcal{P}^{*}\}=\mathop{\arg\max}\limits_{\mathcal{A}, \mathcal{P}} \overline{\mathcal{~U}^b}(q^*, \gamma, n_b, \mathcal{A}, \mathcal{P}) \vspace{-1ex} \\[3mm]
q^*=\mathop{\arg\max}\limits_{q} \overline{~\mathcal{U}^b}(q, \gamma, n_b, \mathcal{A}, \mathcal{P}) 
\end{cases}\\
&\text{s.t.}\notag \\
& C1: 1\le \mathcal{A}\le V,  \notag \\
& C2: p^{min}\le \mathcal{P}\le p^{max}, \vspace{-2ex}  \notag \\
& C3: 0<q\le q^{max}, \notag \\
& C4: \mathcal{R}^s(n_l, \mathcal{A}, \mathcal{P})\le \lambda^s_2 , \notag \\
& C5: \mathcal{R}^b(q^*, \gamma, n_b, \mathcal{A}, \mathcal{P})\le \lambda^b_2 ,\notag 
\end{align}
where $q^*$ stands for the buyer's feasible transmission power while discussing the final contract term with the seller. Constraints $C1$, $C2$, and $C3$ limit the amount of trading resources and price as well as the UAV's transmission power, respectively. Notably, $p^{max}=p^{min}+\kappa \Delta p, $ where $\kappa $ is a positive integer and $\Delta p$ stands for price granularity. Constraints $C4$ and $C5$ respectively represent the seller's and the buyer's acceptable risk tolerance. It is difficult to solve $\bm{\mathcal{F}_1}$ directly owing to that $\mathcal{A}^*, \mathcal{P}^*$, and the transmission power $q^*$ are coupled with each other, and all need to be optimized. Moreover, the buyer's expected utility given in (11) represents a monotonic increasing function of $\mathbb{C}_{2}$ under any given $\mathcal{A}$ and $\mathcal{P}$, where the exponential integral operation complicates determination of the convexity of $\mathbb{C}_{2}$ on various values of $q$. 

Let $\bm{T}=\{t_1, t_2, \cdots, t_i, \cdots, t_{|\bm{T}|}\}$ denote the set of trading index in the future; ${n_l}^{(t_i)}$, ${n_b}^{(t_i)}$, and ${\gamma}^{(t_i)}$ indicate the number of local users, the number of buyer's tasks, and the A2G channel quality in trading $t_i$, respectively. Thus, the transmission power optimization problem is formulated as $\bm{\mathcal{F}_2}$ in (18), aiming to maximize the buyer's utility during each trading.
\begin{align}
\label{eq18}
& \bm{\mathcal{F}_2}:q^{**}(t_i)=\mathop{\arg\max}\limits_{q}~ \mathcal{U}^b(q, {\gamma}^{(t_i)}, {n_b}^{(t_i)}, \mathcal{A}^*, \mathcal{P}^*), \forall t_i\in \bm{T} ~~~~~~ \text{s.t.}~~C3,
\end{align}
where $q^{**}(t_i)$ denotes the optimal transmission power during each trading under given ${\gamma}^{(t_i)}$ and ${n_b}^{(t_i)}$, based on the predetermined forward contract. Meanwhile, $ \bm{\mathcal{F}_2}$ represents a non-convex optimization problem that complicates the solution design. To solve the above-mentioned problems, we first propose an efficient bilateral negotiation algorithm to facilitate the players' consensus regarding the final term of the forward contract. Then, we investigate a practical power optimization algorithm, through which, the buyer's utility can be maximized by obtaining the optimal transmission power during each trading. 

\begin{algorithm*}[b!]
\setstretch{0.9} 
{\small
\caption{Forward contract design via bilateral negotiation where buyer determines the final contract term}
\label{alg1}
\SetKwInOut{Input}{Input}\SetKwInOut{Output}{Output}
\SetKwInput{KwResult}{Initialization}

\Input{$r_l$, $ V$, $M$, $ p_l$, $ \lambda^s_1$, $\lambda^s_2$, $N$, $\tau ^b$, $ \tau ^s$, $D$, $W$, $q^{max}$, $\varepsilon _1$, $\varepsilon _2$, $\ell$, $\omega _1$, $\omega _2$, $\mathcal{U}^{min}$, $\lambda^b_1$, $ \lambda^b_2$, $p^{min}$, $p^{max}$, $\Delta p$}

\Output{$ \mathcal{A}^*, \mathcal{P}^*$}

\KwResult{$\bm{\mathcal{C}}\leftarrow \varnothing$, \% the candidate contract term set\newline
$n\leftarrow 1$, \% the quotation round \newline
$\mathcal{P}^n\leftarrow p^{min}$, \% for the first iteration, the seller sets the unit trading price to $p^{\min}$  }

\While{$\mathcal{P}^n\le p^{max}$}{

 the seller decides the acceptable range $\bm{S^n}$ of the amount of trading resources while meeting $C1$ and $C4$,  

\If{$\bm{S^n}=\varnothing $}{

 go to Step 17, \% the seller raises the price and starts another round of quotation  

}

 the buyer decides the acceptable range $\bm{B^n}$ of the amount of trading resources while meeting $C1$ and $C6$,  

\If{$\bm{B^n}=\varnothing $}{

 go to Step 20, \% the seller no longer raises the price from iteration $n$ when $\bm{B^n}=\varnothing $  

}

\If{$\bm{S^n}\bigcap \bm{B^n}\neq \varnothing$}{

 $\mathcal{A}^n\leftarrow \mathop{\arg\max}\limits_{\mathcal{A}} {\overline{\mathcal{U}^s}(n_l, \mathcal{A}, \mathcal{P}^n)}$, $ \mathcal{A}\in \bm{S^n}\bigcap \bm{B^n}$ \textbf{\%} the seller chooses the amount of resources that maximizes the seller's expected utility from set $\bm{S^n}\bigcap \bm{B^n}$  

 $\bm{\mathcal{C}}\leftarrow \bm{\mathcal{C}}\bigcup \{\mathcal{A}^n, \mathcal{P}^n\}$, \% put the candidate contract term into $\bm{\mathcal{C}}$  

\Else{go to Step 20, \% the seller no longer raises the price from iteration $n$ when $\bm{S^n}\bigcap \bm{B^n}=\varnothing $}

}

 $n\leftarrow n+1$,  

 $\mathcal{P}^n\leftarrow \mathcal{P}^{n-1}+\Delta p$, \% the seller raises the price and starts another round of quotation  

} 

\If{$\bm{\mathcal{C}}\neq \varnothing$}{

 $\{\mathcal{A}^*, \mathcal{P}^*\}\leftarrow \mathop{\arg\max}\limits_{\mathcal{A}, \mathcal{P}} {\overline{\mathcal{U}^b}(q^{max}, \gamma, n_b, \mathcal{A}, \mathcal{P})}$, $\{\mathcal{A}, \mathcal{P}\}\in \bm{\mathcal{C}$}, \textbf{\%} the buyer chooses the pair of $\mathcal{A}, \mathcal{P}$ from $\bm{\mathcal{C}}$ that maximizes the buyer's expected utility as the final term for the forward contract  

\Else{the trading fails,  }

}

end negotiation. 
}
\end{algorithm*}

\begin{algorithm}[t!]
\setstretch{0.9} 
{\small
\caption{Forward contract design via bilateral negotiation where seller determines the final contract term}
\label{alg2}
\setcounter{AlgoLine}{8}
$\mathcal{A}^n\leftarrow \mathop{\arg\max}\limits_{\mathcal{A}} {\overline{\mathcal{U}^b}(q^{max}, \gamma, n_b, \mathcal{A}, \mathcal{P}^n)}$, $ \mathcal{A}\in \bm{S^n}\bigcap \bm{B^n}$

\setcounter{AlgoLine}{15}
$\{\mathcal{A}^*, \mathcal{P}^*\}\leftarrow \mathop{\arg\max}\limits_{\mathcal{A}, \mathcal{P}} {\overline{\mathcal{U}^s}(n_l, \mathcal{A}, \mathcal{P})}$, $\{\mathcal{A}, \mathcal{P}\}\in \bm{\mathcal{C}}$
}
\end{algorithm}

\section{Proposed Futures-based Fast Resource Trading Mechanism}
\begin{figure*}[t!]
\centering
\subfigure[]{\includegraphics[width=.325\linewidth]{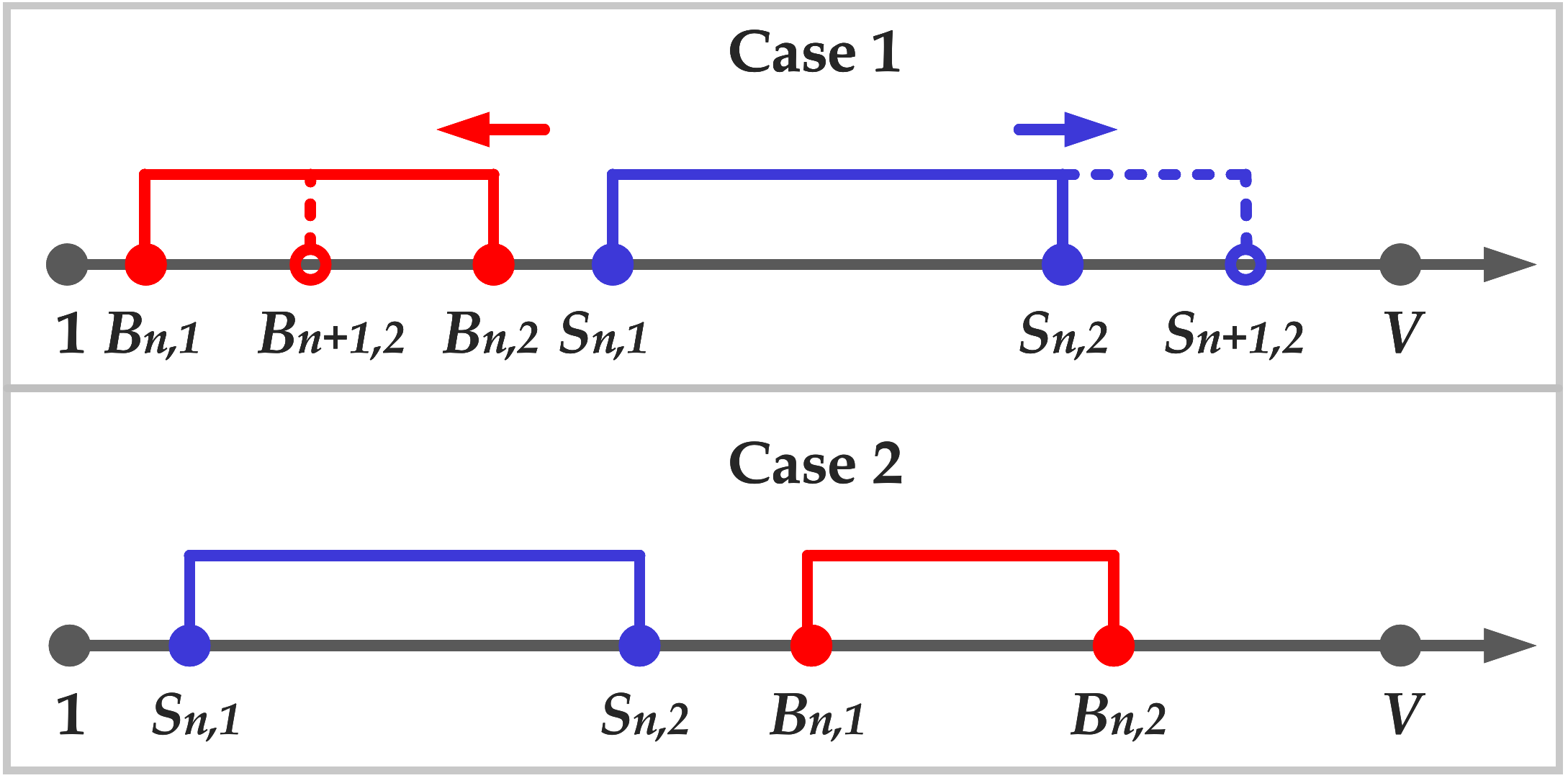}}~
\subfigure[]{\includegraphics[width=.325\linewidth]{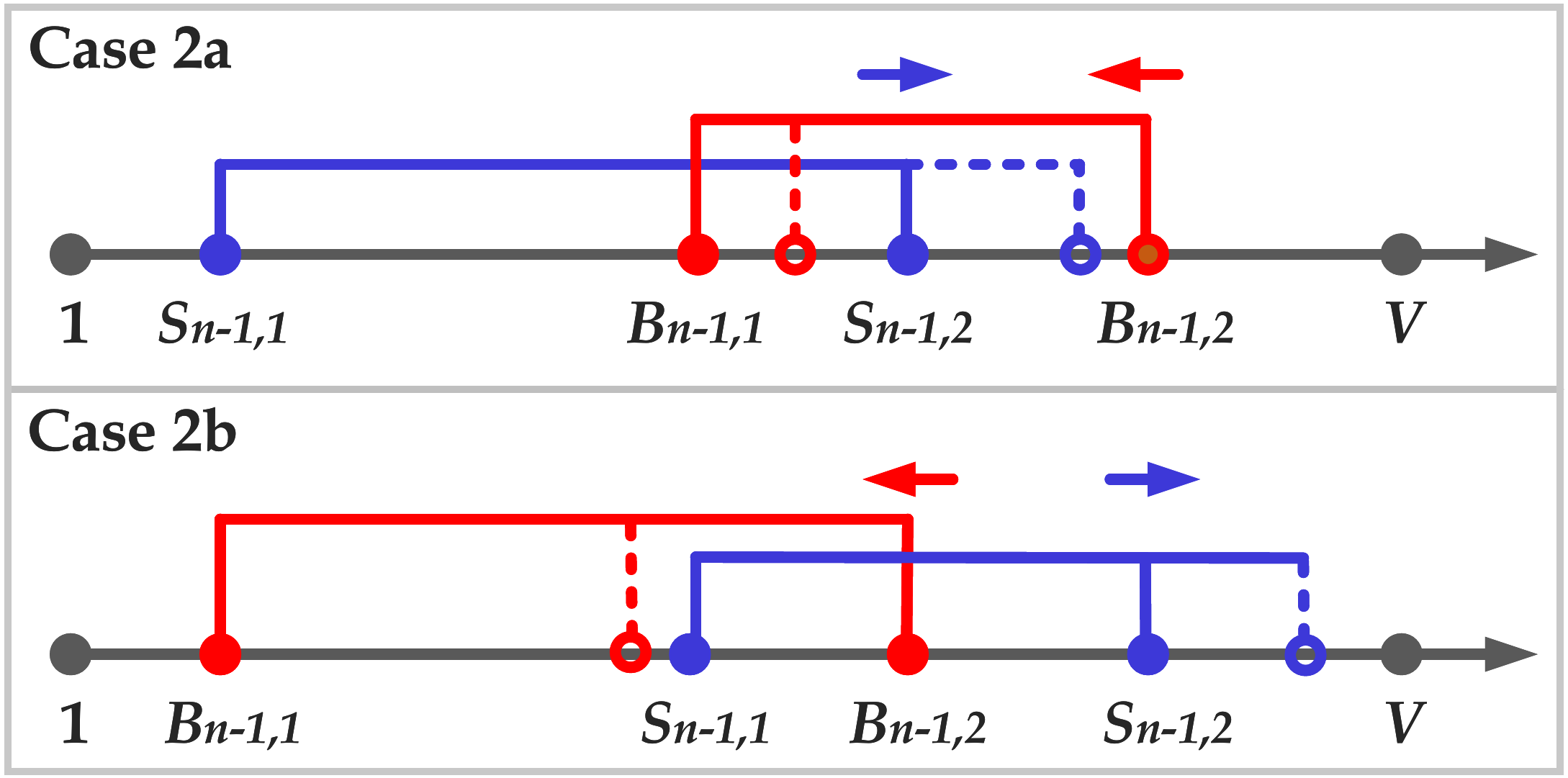}}~
\subfigure[]{\includegraphics[width=.325\linewidth]{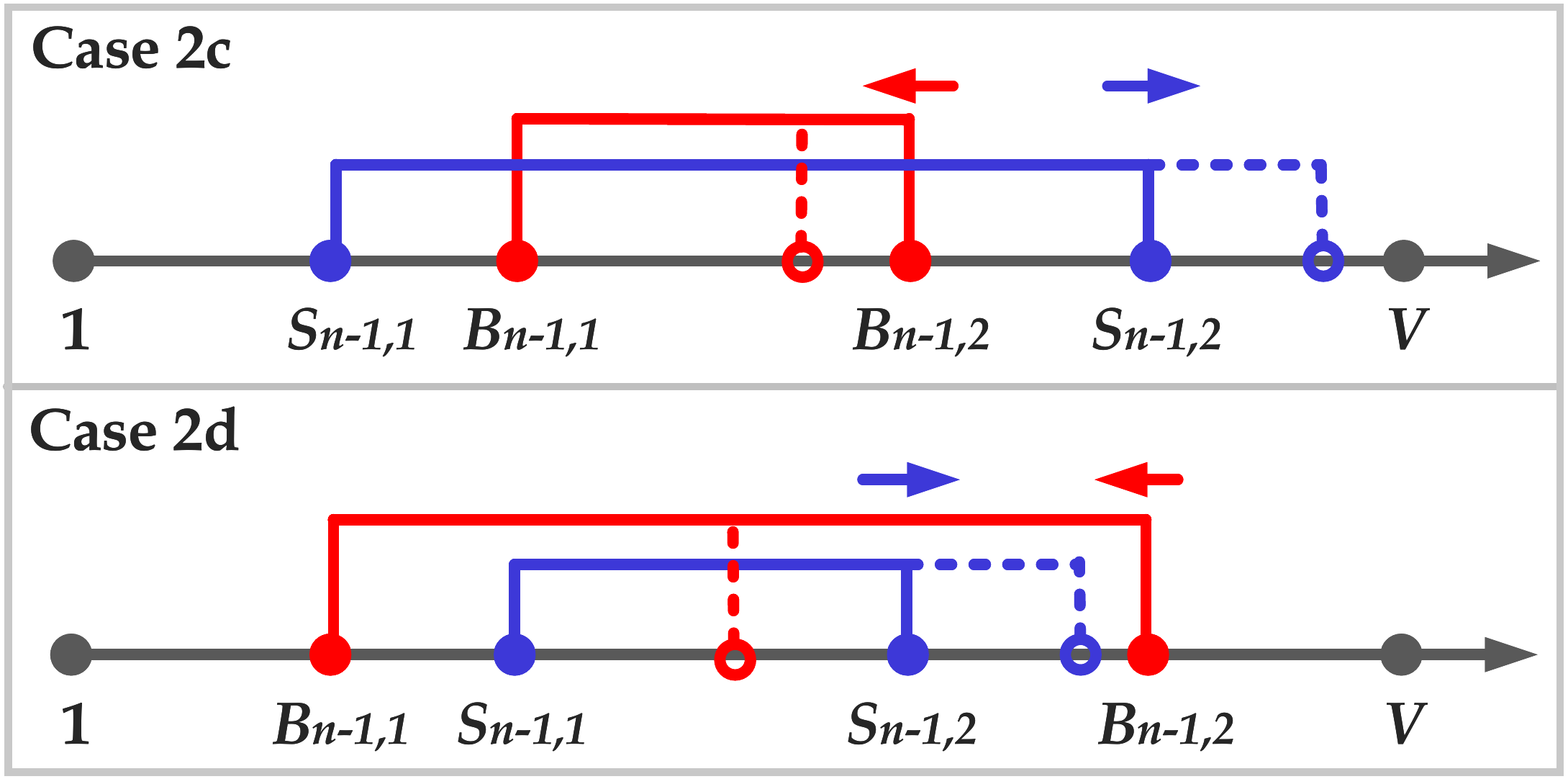}}
\vspace*{-.6em}
\caption{Diagram associated with \textbf{Proposition 1}.}
\label{fig3}
\end{figure*}

\subsection{Forward Contract Design via Bilateral Negotiation}
\noindent
As the uncertainties of $n_b $ and $\gamma $ present challenges when the buyer has to ascertain feasible transmission power during the design of the forward contract, we rewrite the MOO problem $ \bm{\mathcal{F}_1}$ into $\bm{\mathcal{F}_3}$ as given in (19) by letting $q^*=q^{max}$.
\begin{align}
\label{eq19}
&\bm{\mathcal{F}_3}: 
\begin{cases}
\{\mathcal{A}^{*}, \mathcal{P}^{*}\}=\mathop{\arg\max}\limits_{\mathcal{A}, \mathcal{P}} \overline{~\mathcal{U}^s}(n_l, \mathcal{A}, \mathcal{P}) \\ 
\{\mathcal{A}^{*}, \mathcal{P}^{*}\}=\mathop{\arg\max}\limits_{\mathcal{A}, \mathcal{P}} \overline{~\mathcal{U}^b}(q^{max}, \gamma, n_b, \mathcal{A}, \mathcal{P}) 
\end{cases}\\
&\text{s.t.}~~C1, C2, C4, \notag \\
&C6: \mathcal{R}^b(q^{max}, \gamma, n_b, \mathcal{A}, \mathcal{P})\le \lambda^b_2.\notag
\end{align}
Owing to the information privacy, traditional~methods~for solving the MOO problem (e.g., weighted sum method~\cite{40}, weighted metric method~\cite{41}, and multi-objective genetic algorithms~\cite{42}) are difficult to implement in this paper. Consequently, bilateral negotiation is introduced as an efficient approach whereby the two players negotiate the unit price and the amount of resources to be traded under the forward contract in an iterative manner. Algorithm~1 depicts the detailed logic of the proposed negotiation, in which the seller chooses the candidate contract terms (step 9, Algorithm~1), and the buyer determines the final contract term (step 16, Algorithm~1). Specifically, the seller starts with $p^{min}$ for the first quotation round. In the $n^{\text{th}}$ quotation, the seller sets a price $\mathcal{P}^n$ and decides the acceptable range of trading resources denoted by $\bm{S}^{n}$ while meeting constraints $C1$ and $C4$ (step 2, Algorithm~1); if $\bm{S^n}=\varnothing$, the seller directly raises the price and initiates the next round of negotiation (steps 3 and 4, Algorithm~1) because the current price might not meet its risk tolerance. Under a given $\mathcal{P}^n$, the buyer determines an affordable resource-trading range, denoted by $\bm{B^{n}}$, that satisfies constraints $C1$ and $C6$ (step 5, Algorithm~1). The quotation procedure will be ended if $\bm{B^n}=\varnothing $ (steps 6 and 7, Algorithm~1; see proof in \textbf{Proposition 1}). Then, if the two ranges overlap, the seller outlines a candidate contract term by choosing the amount of resources that maximizes the seller's expected utility (steps 9 and 10, Algorithm~1). Notably, to facilitate an efficient negotiation, the seller will no longer quote if $\bm{S^n}\bigcap \bm{B^n}=\varnothing$ (steps 11 and 12, Algorithm~1; see proof in \textbf{Proposition 1}). 

When all the candidate contract terms are settled, the buyer determines the final contract term $\{\mathcal{A}^{*}, \mathcal{P}^{*}\}$ by choosing the one that maximizes the buyer's expected utility (steps 15 and 16, Algorithm~1). Otherwise, if there is no candidate contract term, then futures-based trading fails.

Note that cases in which the final contract term is determined by the seller or the buyer can differ; thus, we introduce Algorithm~2 where all steps are the same as in Algorithm~1 except for steps 9 and 16. Specifically, the buyer chooses the candidate contract terms (step 9, Algorithm~2), while the seller determines the final contract term (step 16, Algorithm~2).

\begin{prop}
\label{prop1}
To facilitate an efficient negotiation, the seller will no longer quote from the $n^{\text{th}}$ round when either of the following conditions applies:

\noindent
Condition 1: $\bm{B^n}=\varnothing$; 

\noindent
Condition 2: $\bm{S^n}\bigcap \bm{B^n}=\varnothing $ ($\bm{S^n}\neq \varnothing$ and $\bm{B^n}\neq \varnothing$).

\noindent
Namely, raising the price by $\mathcal{P}^{n+1}\leftarrow \mathcal{P}^n+\Delta p$ will not elicit any more candidate contract terms.
\end{prop}
\begin{proof}
When $\bm{B^n}=\varnothing$, we have $\mathcal{R}^b(q^{max}, \gamma, n_b, \mathcal{A}, \mathcal{P}^n)\le \mathcal{R}^b(q^{max}, \gamma, n_b, \mathcal{A}, \mathcal{P}^n+\Delta p)$ owing to the monotonic non-decreasing property of $\mathcal{R}^b$, which results in $\bm{B^{n+1}}=\varnothing$.

When $\bm{S^n}\bigcap \bm{B^n}=\varnothing$ ($\bm{S^n}\neq \varnothing$ and $\bm{B^n}\neq \varnothing$), suppose that under given $\mathcal{P}^n$ in round $n$, the seller's acceptable range of trading resources is denoted by $\bm{S^n}=\{S_{n, 1}, S_{n, 1}+1, S_{n, 1}+2, \cdots, S_{n, 2} \}$; that of the buyer is represented by $\bm{B^n}=\{B_{n, 1}, B_{n, 1}+1, B_{n, 1}+2, \cdots, B_{n, 2}\}$. Correspondingly, we consider the following two cases (see Fig.~2):

\noindent
$\bullet$ \textit{Case 1} ($S_{n, 1}>B_{n, 2}$): In round $n+1$, raising the price by $\mathcal{P}^{n+1}=\mathcal{P}^n+\Delta p$ will bring a lower $\mathbb{C}_{1}$, and a larger $\mathbb{C}'_{3}$. Consequently, the seller can appropriately increase $S_{n, 2}$ to $S_{n+1, 2}$ while satisfying constraint $C4$, based on the monotonic non-decreasing property of $\mathcal{R}^s$. On the contrary, the buyer will reduce $B_{n, 2}$ to $B_{n+1, 2}$ to meet constraint $C6$. Apparently, the acceptable ranges of trading resources of the two players will no longer overlap from round $n$.

\noindent
$\bullet$ \textit{Case 2} ($S_{n, 2}<B_{n, 1}$): We consider the following sub-cases in round $n-1$ (see Cases 2a, 2b, 2c, and 2d in Fig.~2). In Case 2a, we have $S_{n-1, 1}<B_{n-1, 1}<S_{n-1, 2}<B_{n-1, 2}$; in Case 2b, we have $B_{n-1, 1}<S_{n-1, 1}<B_{n-1, 2}<S_{n-1, 2}$; in Case 2c, we have $S_{n-1, 1}<B_{n-1, 1}<B_{n-1, 2}<S_{n-1, 2}$; and in Case 2d, we have $B_{n-1, 1}<S_{n-1, 1}<S_{n-1, 2}<B_{n-1, 2}$. Similarly, raising the unit price by $\mathcal{P}^n=\mathcal{P}^{n-1}+\Delta p$ enables the seller to increase $S_{n-1, 2}$ while forcing the buyer to reduce $B_{n-1, 2}$. Apparently, Case 2 will never happen.

In conclusion, the seller will no longer quote when either $\bm{S^n}\bigcap \bm{B^n}=\varnothing $~($\bm{S^n}\neq \varnothing$ and $\bm{B^n}\neq \varnothing$), or $\bm{B^n}=\varnothing$, to facilitate an efficient negotiation during forward contract design.
\end{proof}

\subsection{Transmission Power Optimization}
\noindent
During each future trading, the buyer can adjust its transmission power to obtain better utility. Note that the proposed algorithm can be applied to all the trading, we ignore the label ``$(t_i)$'' for analytical simplicity hereafter. Correspondingly, the power optimization problem $ \bm{\mathcal{F}_2}$ is reformulated by $ \bm{\mathcal{F}_4}$ as given in (20).
\vspace*{-0.4\baselineskip}  
\begin{align}
\label{eq20}
& \bm{\mathcal{F}_4}: q^{**}=\mathop{\arg\max_{q}~\mathcal{U}^b(q, \gamma, n_b, \mathcal{A}^*, \mathcal{P}^*)} ~~~~~~\text{s.t.}~C3
\end{align}

To facilitate analysis, an optimization problem $\bm{\mathcal{F}_{5}}$ equivalent~to~$\bm{\mathcal{F}_{4}}$ is applied in (21) under given $\mathcal{A}^*$ and $\mathcal{P}^*$:
\vspace*{-0.4\baselineskip}  
\begin{align}
\label{eq21}
& \bm{\mathcal{F}_{5}}: q^{**}=\mathop{\arg\min}_{q} f(q, \gamma)~~~~~~\text{s.t.}~C3,
\end{align}
where $f(q, \gamma)={(1+\omega _2q)}/{\log _{2}(1+q\gamma)}$, the relevant second derivative of which is given by (22), depicting a non-convex function that further complicates the problem:
\begin{align}
\label{eq22}
& \frac{\partial ^2f(q, \gamma)}{\partial ^2q}={\frac{\gamma (-\omega _2\gamma q\ln (1+q\gamma)+2\omega _2\gamma -2\omega _2\ln (1+q\gamma)+\gamma \ln (1+q\gamma)+2\gamma)}{{(\ln 2)}^2\times {\log _2(1+q\gamma)}^3{\times (1+q\gamma)}^2}}.
\end{align}
Thus, the method of changing variables is considered by applying $\beta={1}/{\log _{2}(1+q\gamma)}$, which enables an equivalent and convex optimization problem $ \bm{\mathcal{F}_6}$ given in (23):
\begin{align}
\label{eq23}
&\bm{\mathcal{F}_6}: {\beta}^* =\mathop{\arg\min}_{\beta} h(\beta, \gamma)\\
 & \text{s.t. }C7: \beta \ge \frac{1}{\log _2(1+\gamma q^{max})},\notag  
\end{align}
where $h(\beta, \gamma)=\beta +\frac{\omega _2\beta \left(2^{\frac{1}{\beta}}-1\right)}{\gamma}$ with the second derivative $\frac{\partial ^2h(\beta, \gamma)}{\partial ^2\beta} =\frac{{(\ln 2)}^2\times \omega _2{\times 2}^{\frac{1}{\beta}}}{\gamma {\beta}^3}>0$. Consequently, $ \bm{\mathcal{F}_6}$ can be solved by letting $\frac{\partial h(\beta, \gamma)}{\partial \beta} =0$ while meeting constraint $C7$. Thus, we can calculate $q^{**}$ as shown in (24), where $\textbf{W}(\cdot)$ denotes the Lambert W function~\cite{43}; the relevant derivation can be found in the Appendix.
\begin{align}
\label{eq24}
& q^{**}=
\begin{cases}
\vspace{-0.07in}
\frac{{\left(2^{\frac{\textbf{W}\left(\frac{\gamma -\omega _2}{\rm{e}\times \omega _2}\right)+1}{\ln 2}}-1\right)}}{\gamma} , \text{~if~}\frac{\ln 2}{\textbf{W}\left(\frac{\gamma -\omega _2}{\rm{e}\times \omega _2}\right)+1}>\frac{1}{\log _2(1+\gamma q^{max})} \vspace{-0.5ex} \\ 
q^{max}, \text{~~~~~~~~~~~~~~~~~~otherwise} 
\end{cases}
\end{align}

%
%

\begin{figure*}[b!]
\centering
\subfigure[]{\includegraphics[width=.244\linewidth]{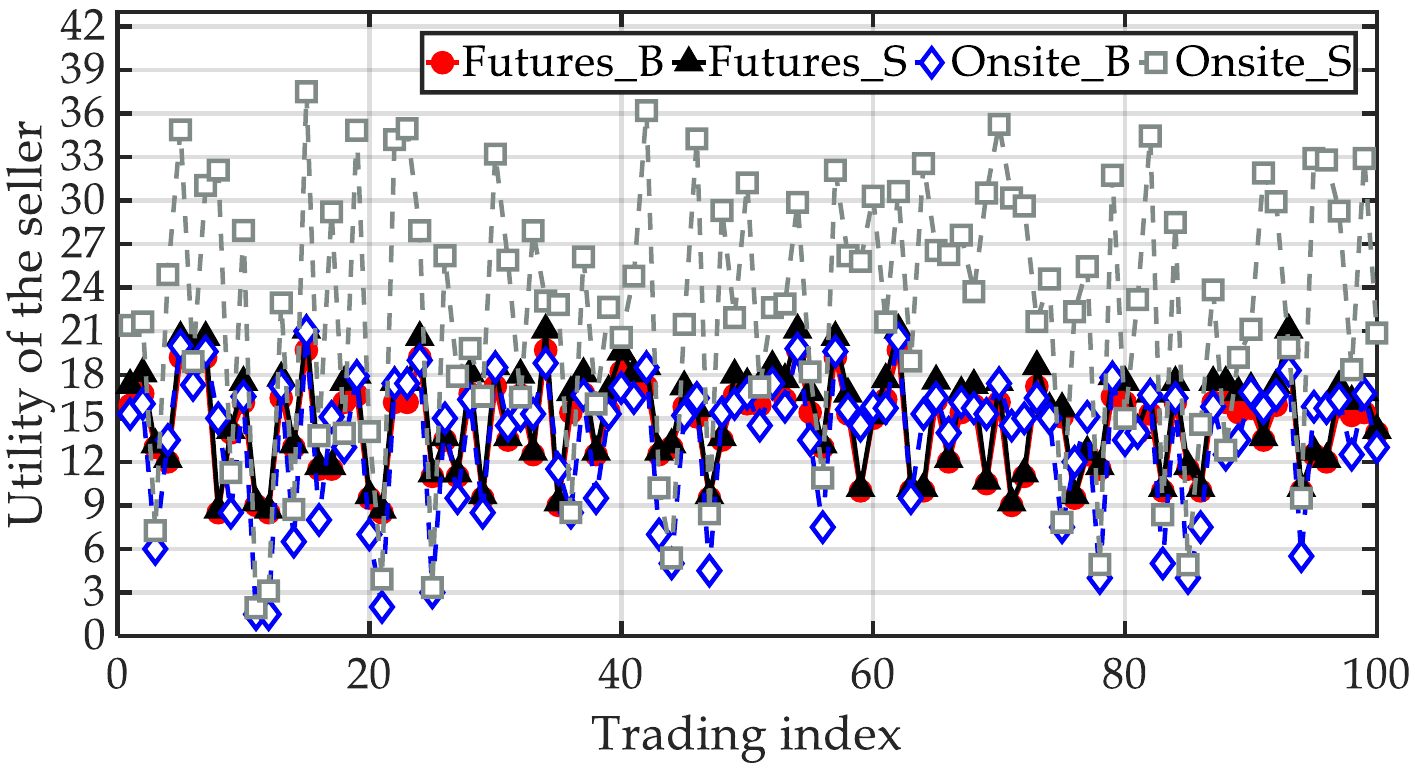}}
\subfigure[]{\includegraphics[width=.244\linewidth]{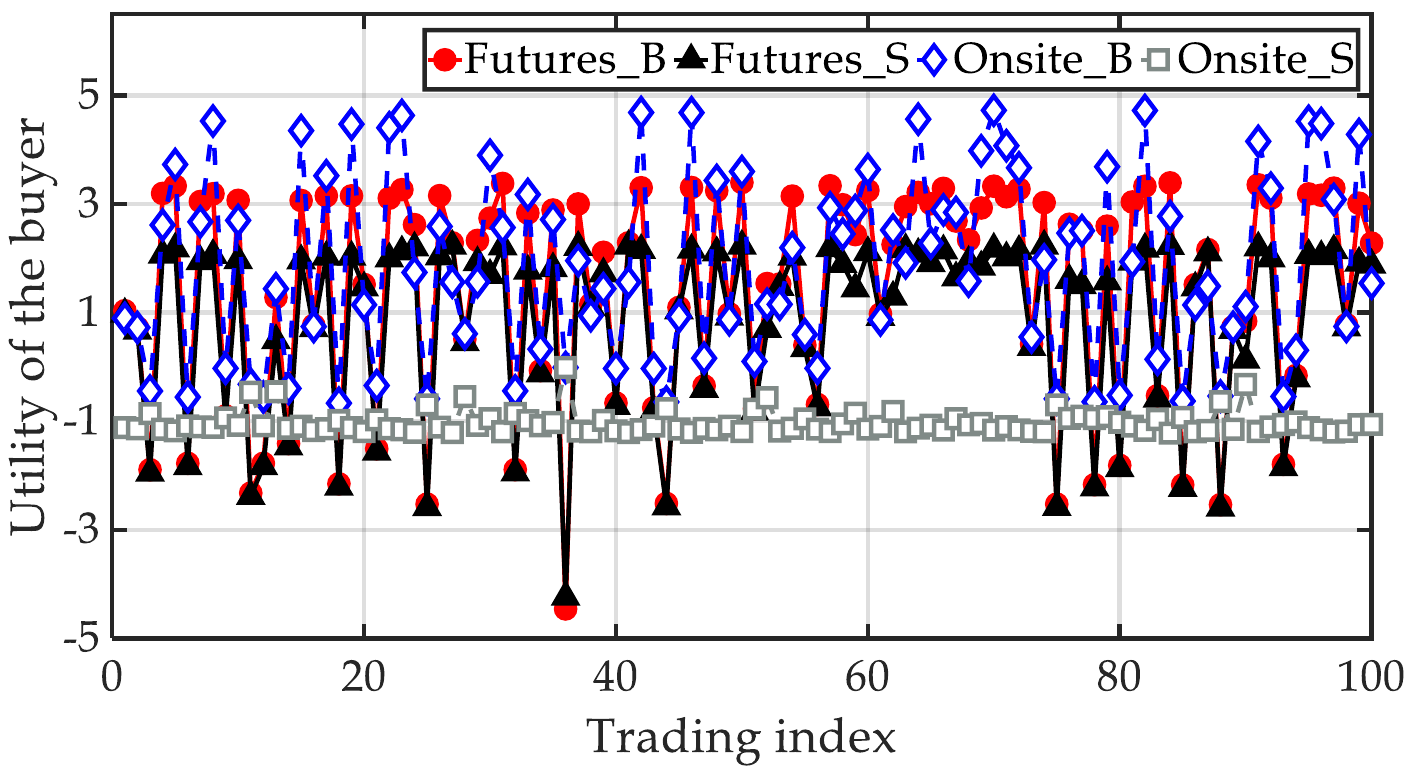}}
\subfigure[]{\includegraphics[width=.244\linewidth]{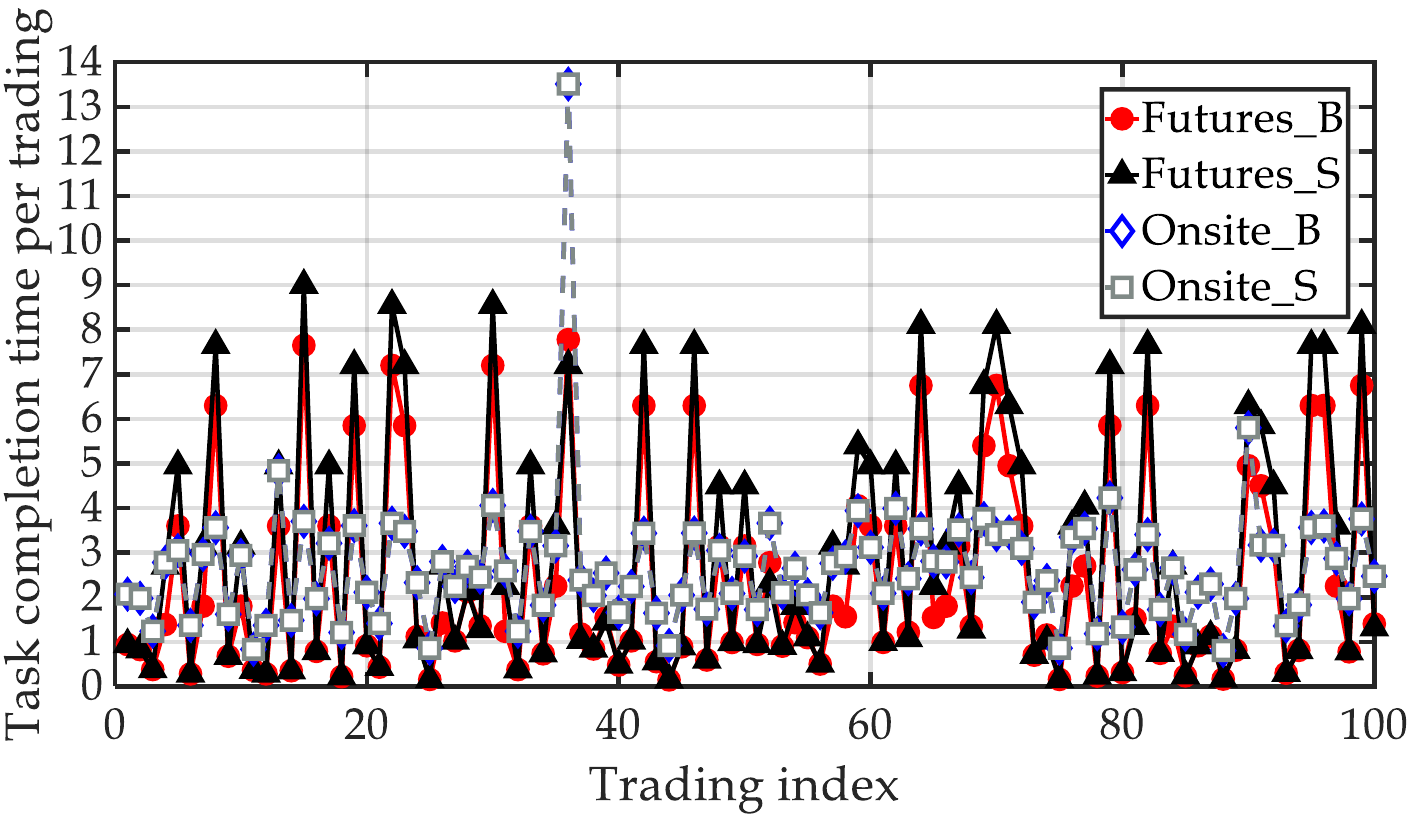}}
\subfigure[]{\includegraphics[width=.244\linewidth]{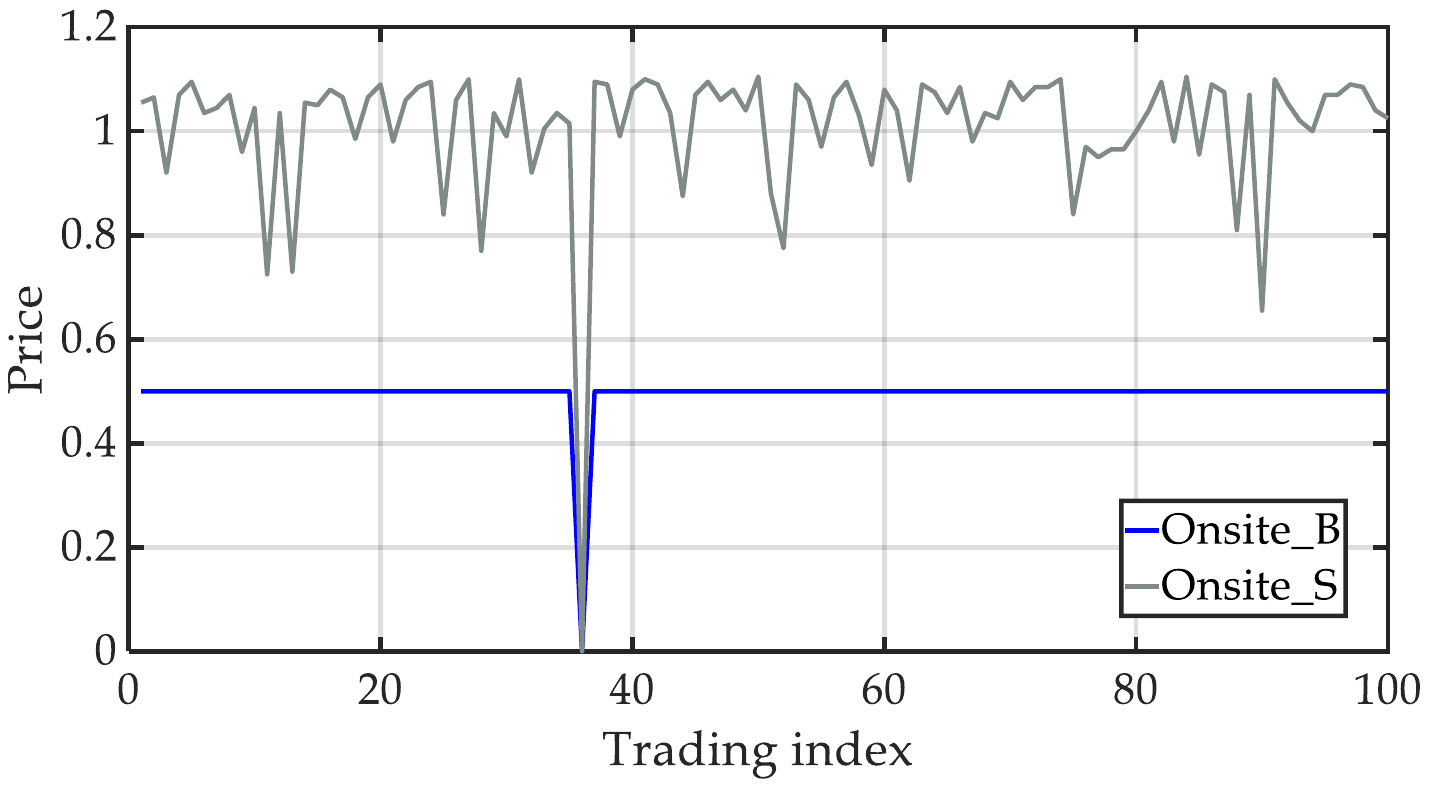}}\\
\vspace*{-.6em}
\subfigure[]{\includegraphics[width=.244\linewidth]{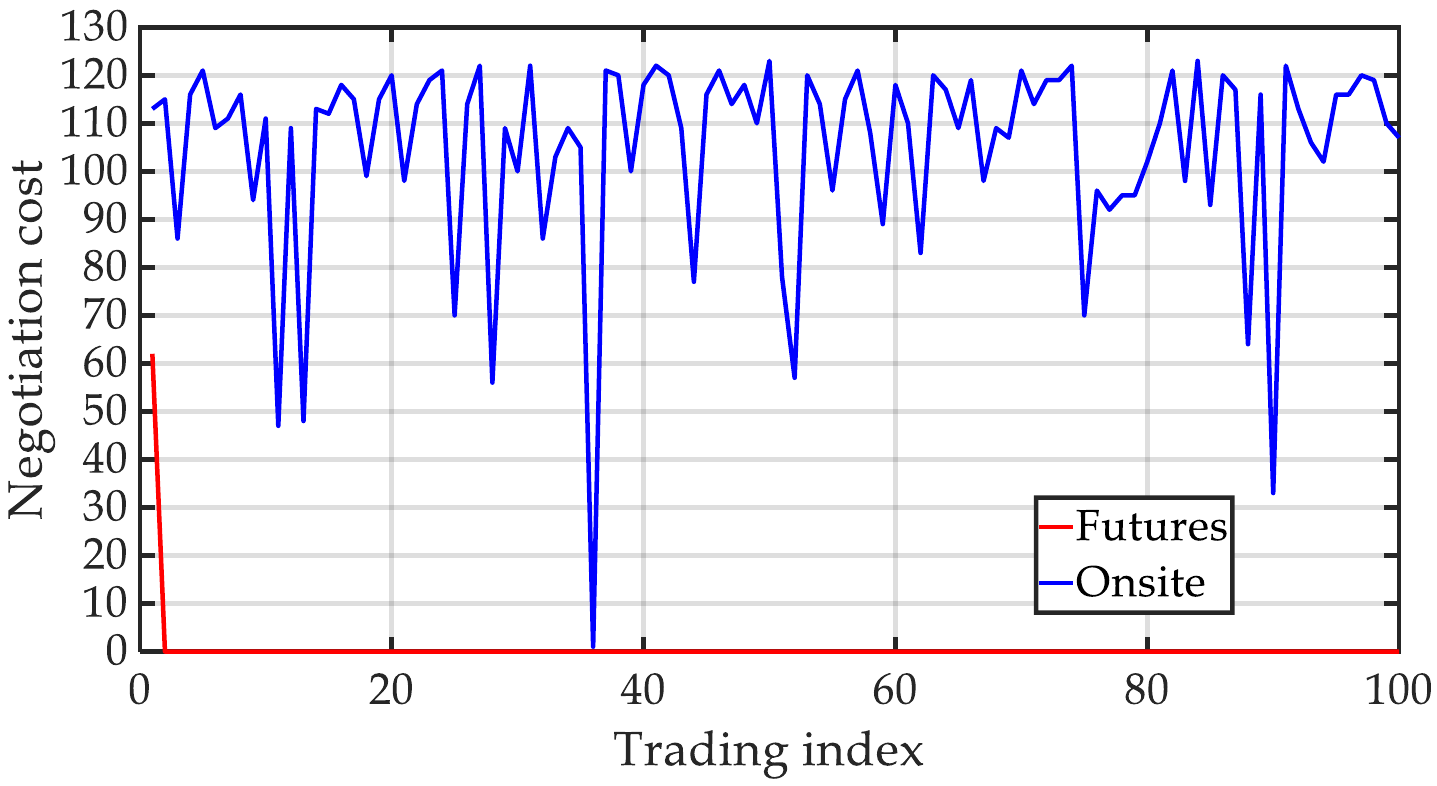}}
\subfigure[]{\includegraphics[width=.244\linewidth]{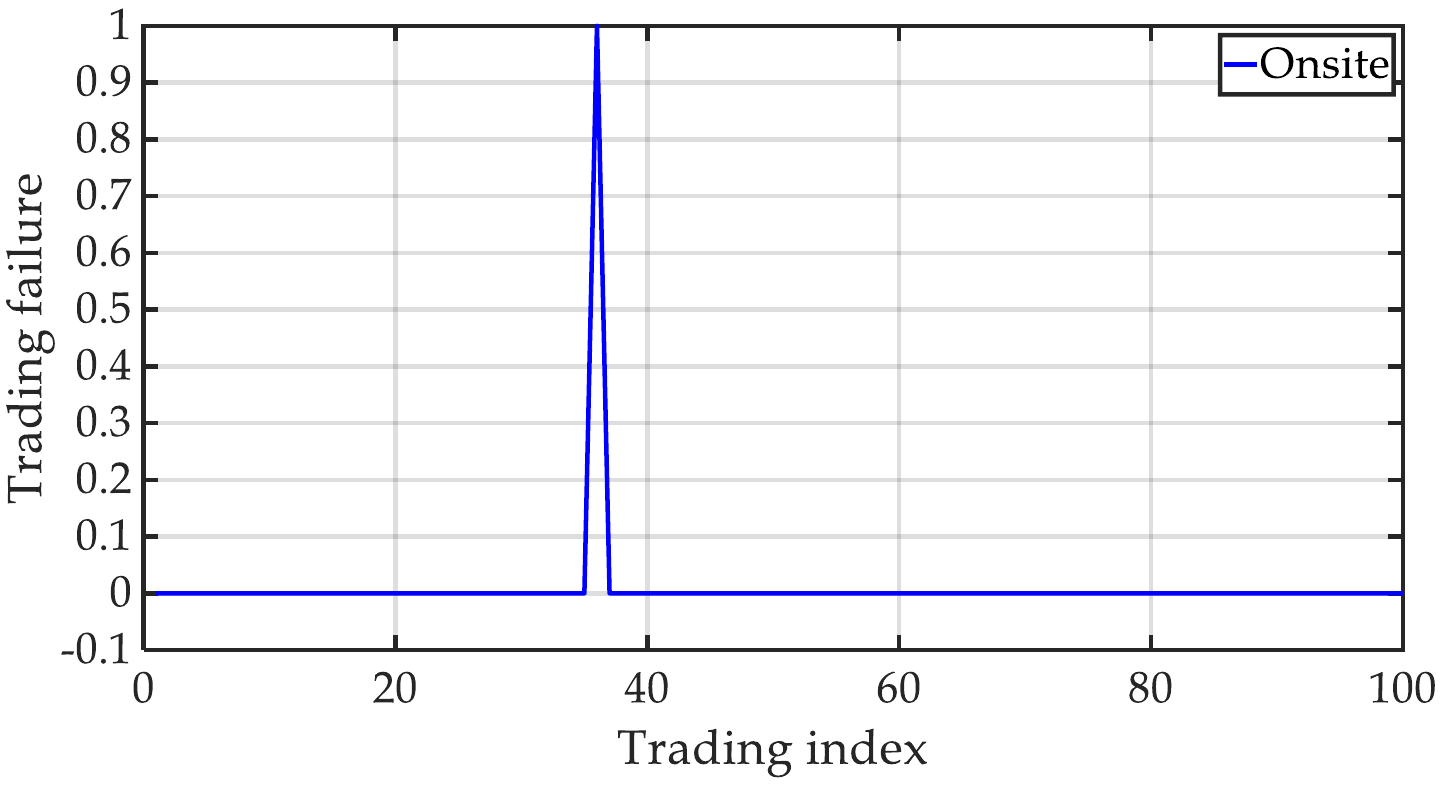}}
\subfigure[]{\includegraphics[width=.244\linewidth]{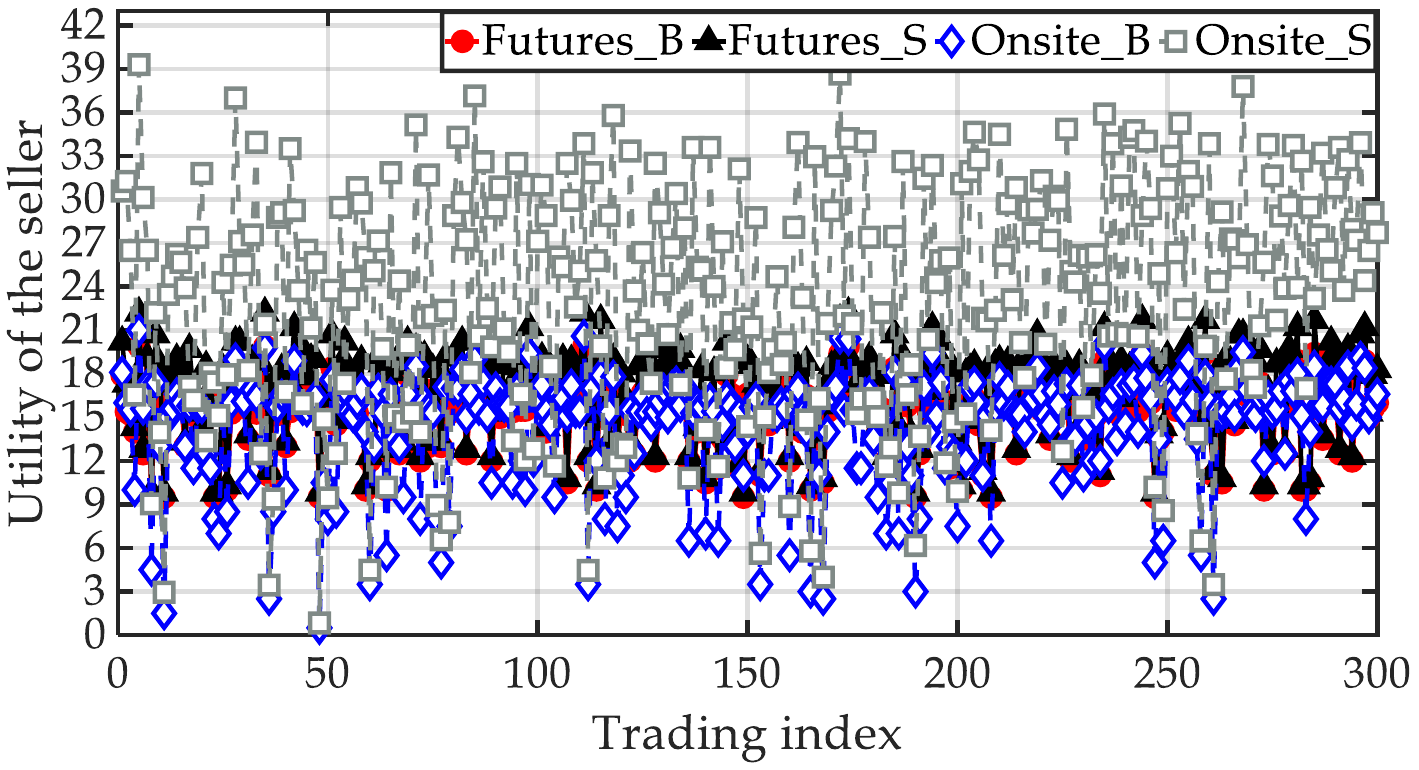}}
\subfigure[]{\includegraphics[width=.244\linewidth]{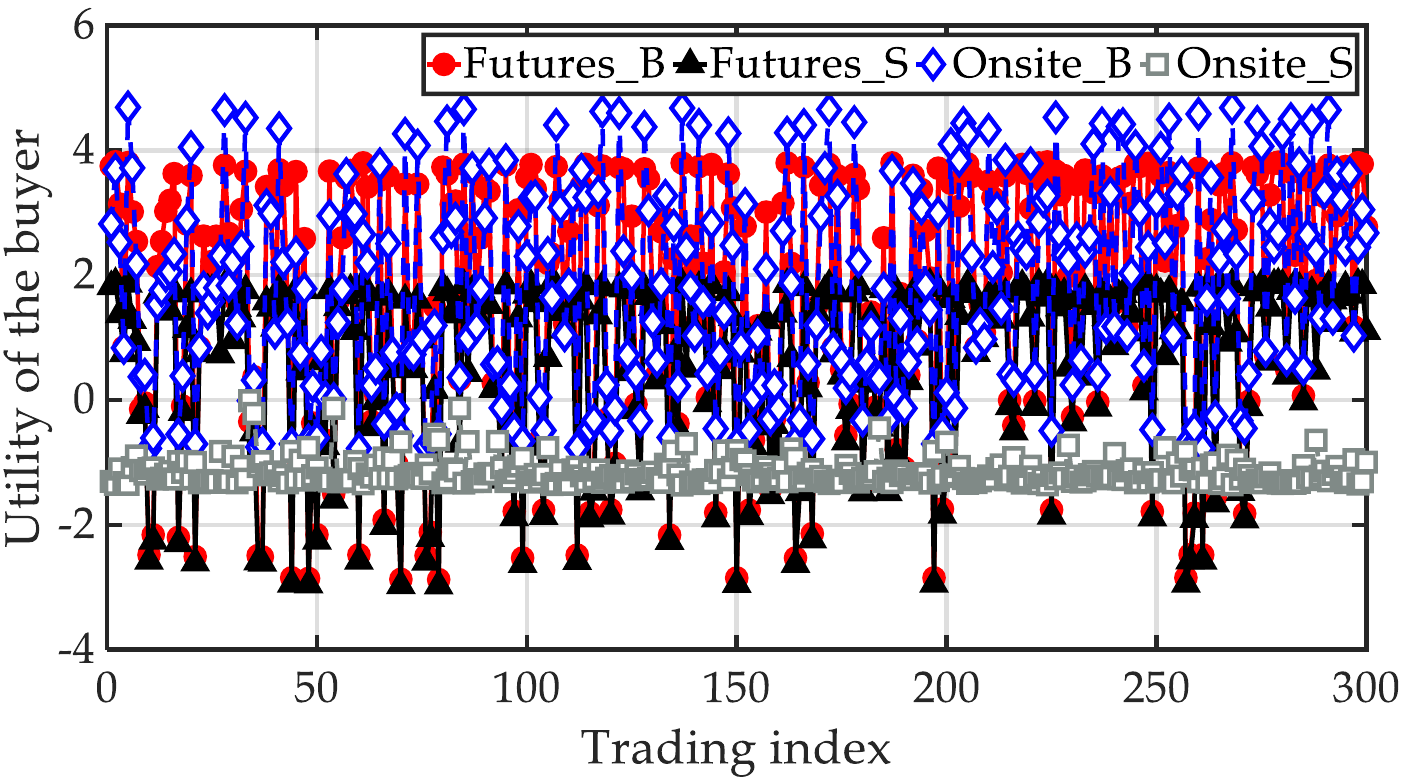}}\\
\vspace*{-.6em}
\subfigure[]{\includegraphics[width=.244\linewidth]{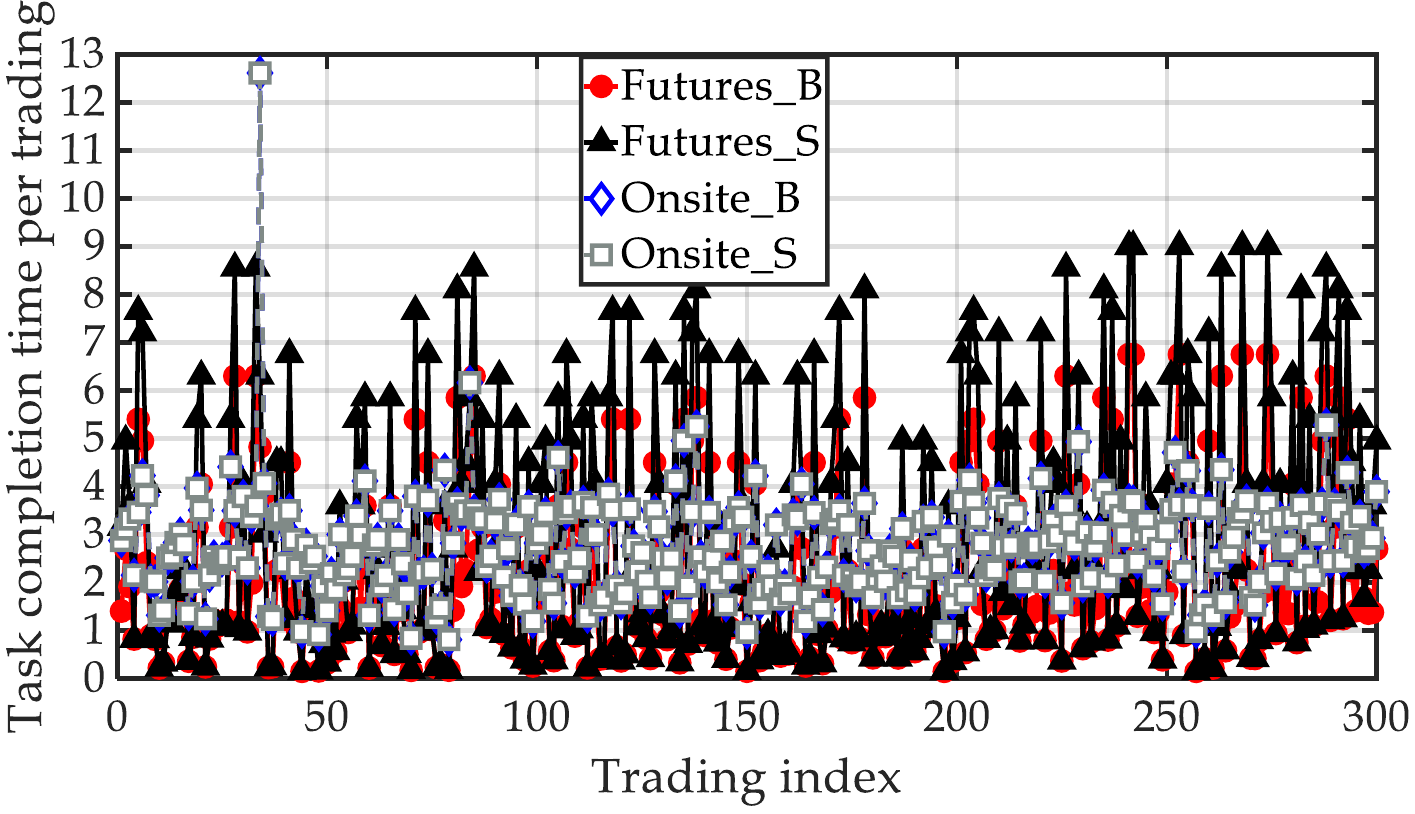}}
\subfigure[]{\includegraphics[width=.244\linewidth]{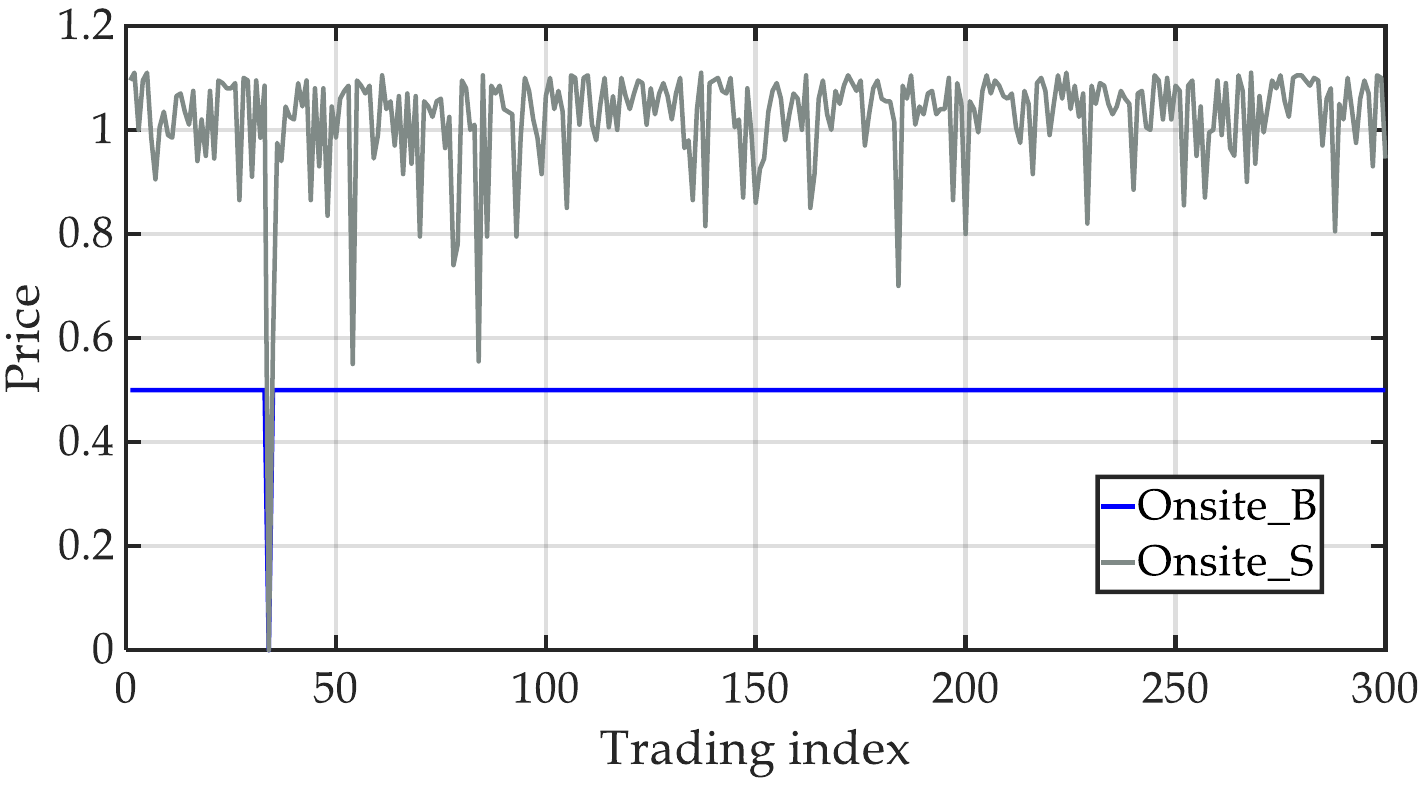}}
\subfigure[]{\includegraphics[width=.244\linewidth]{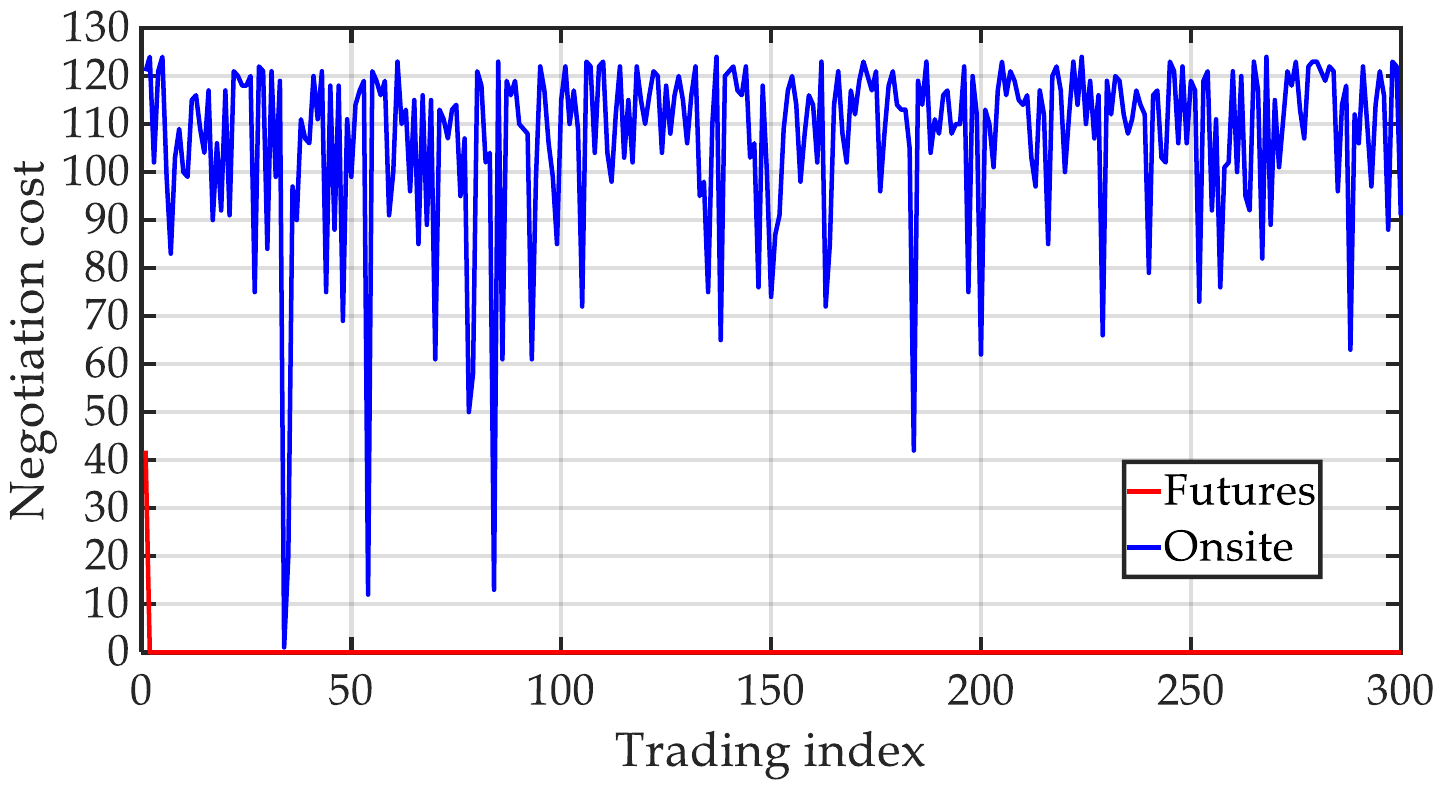}}
\subfigure[]{\includegraphics[width=.244\linewidth]{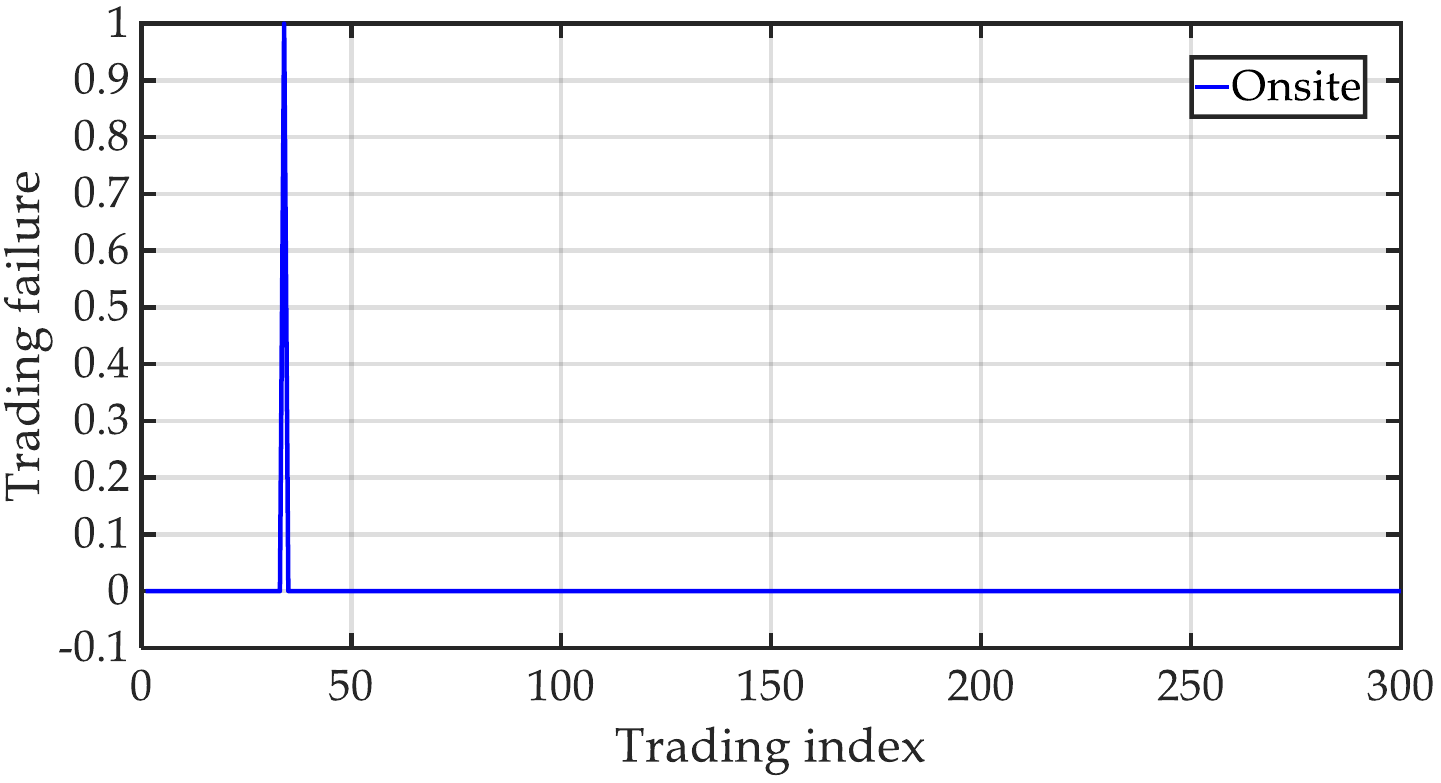}}
\caption{Performance comparison and evaluation of players' utilities, task completion time, resource prices, negotiation costs, and trading failures considering 100 and 300~trading ($\tau ^s=0.08, \tau ^b=0.45$, $\Delta p=0.005 $).}
\label{fig4}
\vspace*{-.6em}
\end{figure*}

\section{Experimental Results}

\noindent
This section presents comprehensive simulation results based on the Monte Carlo method along with performance evaluations illustrating the validity of the proposed futures-based resource trading mechanism,  compared with the onsite trading mechanism. Specifically, in each onsite trading, the two players negotiate an agreement regarding the amount of trading resources and the relevant price that maximizes their utilities, relying on the current $n_l$, $n_b$, and $\gamma $. For notational simplicity, the proposed futures-based resource trading mechanisms for which the buyer and the seller determine the final contract term are abbreviated to ``Futures\_B'' and ``Futures\_S'', respectively (collectively termed ``Futures''). Similarly, let ``Onsite\_B'' and ``Onsite\_S'' represent the onsite trading where the buyer and the seller respectively decide the final trading resources and relevant price (collectively termed ``Onsite''). 

\subsection{Critical Indicators and Parameter Settings}

\noindent
To reveal the validity of the proposed mechanism, several critical indicators applied in this simulation are introduced below:

\noindent
$\bullet$ \textbf{Trading unfairness (UFair)}: Unfairness in a resource trading system is mainly attributable to fluctuating prices~\cite{25,33}. In this simulation, we calculate UFair as the standard~deviation~of prices, where a larger standard~deviation~leads to worse fairness.

\noindent
$\bullet$ \textbf{Negotiation cost (NC) and latency (NL)}: In this simulation, NC is mainly considered as the players' energy and battery consumption (especially for the buyer) while negotiating the trading consensus. Note that NC is difficult to quantize by a numerical value; thus, we take the number of quotation rounds (e.g., $n$ in Algorithm~1) to reflect NC. Moreover, the latency for reaching a trading consensus is calculated as $\text{NL}=n\times t^{NL}$, where $t^{NL}$ indicates the latency per quotation.

\noindent
$\bullet$ \textbf{Task completion time} $t^{Comp}$: The buyer's task completion time in each trading is calculated as (25), via considering the task offloading delay through the A2G wireless link, task execution time, and negotiation latency. Specifically, $(\cdot)^+$ refers to the larger completion time between local computing and edge computing.
\vspace*{-0.4\baselineskip}  
\begin{align}
\label{eq25}
& t^{Comp}={\left(\tau ^b(n_b-{[\mathcal{A}, n_b]}^-), \left(\tau ^s+\frac{D{[\mathcal{A}, n_b]}^-}{W\log _{2}(1+{\gamma q}^{**})}\right)\right)^++\text{NL}} 
\end{align}

\noindent
$\bullet$ \textbf{Trading failures (TFail)}: In this simulation, TFail denotes the number of failed trading. Specifically, when a trading fails, let $\mathcal{A}= \mathcal{P}=q=\mathcal{U}^b=0$. 

\noindent
$\bullet$ \textbf{Buyer's net utility}: Because onsite negotiation may lead to heavy latency, the buyer's utility during each trading is recalculated as the net utility, given by $\mathcal{U}^b\leftarrow \mathcal{U}^b-\text{NL}$. 

\noindent
Major simulation parameters are set as follows: $r_l\in [0.2, 0.4]$, $p_l\in [0.3, 0.5]$, $V\in [30, 35]$, $M\in (30, 40]$, $N\in (30, 40]$, $\lambda^s_1\in [0.95, 1]$, $\lambda^s_2\in [0.3, 0.4]$, $D\in [3, 4]$Mb, $W\in [6, 8]$MHz, $q^{max}\in [500, 1000]$mW, $ \ell ={10}^{-5}$, $\varepsilon _1\in [5, 100]$, $\varepsilon _2\in [300, 400]$, $\omega _1\in [0.3, 0.5]$, $\omega _2\in [0.3, 0.5]$, $\lambda^b_1\in [0.3, 0.4]$, $\lambda^b_2\in [0.2, 0.3]$, $\tau ^s=0.08$s, $\tau ^b\in [0.35, 1.6]$s, and $t^{NL}\in [5, 30]$ms~\cite{44}.

\subsection{Performance~Comparison and Evaluation}

\noindent
Fig.~3 depicts the performance comparison and evaluation of the players' utilities (Fig.~3(a), Fig.~3(b), Fig.~3(g), and Fig.~3(h)), the task completion time (Fig.~3(c) and Fig.~3(i)), the price fluctuation (Fig.~3(d) and Fig.~3(j)), the negotiation cost (Fig.~3(e) and Fig.~3(k)), and trading failures (Fig.~3(f) and Fig.~3(l)), upon having different number of trading. Additionally, Table~\ref{tab1} lists the relevant indicators associated with Fig.~3. As illustrated in Fig.~3(a), Fig.~3(b), Fig.~3(g), and Fig.~3(h), the seller's utility of Futures\_S and Onsite\_S is higher in most of the trading compared to that of Futures\_B and Onsite\_B; more specifically, the seller determines the trading resources and relevant price by maximizing its own benefit. In particular, Onsite\_S sometimes offers the seller greater utility than Futures\_S due to prediction-related uncertainties during each trading. However, Onsite\_S consistently leads to negative utility for the buyer, whereas the proposed Futures\_S achieves far better performance (see Sum($\mathcal{U}^b$) of  Table~\ref{tab1}). 

\begin{table*}[h!]
{\small
\centering
\caption{Evaluations of critical indicators associated with Fig.~3}
\setlength{\tabcolsep}{1.5mm}{
\begin{tabular}{|*{9}{c|}}
\hline
Number of trading & \multicolumn{4}{c|}{100~trading} &  \multicolumn{4}{c|}{300~trading} \\ 
\hline
Algorithm & Futures\_B & Futures\_S & Onsite\_B & Onsite\_S & Futures\_B & Futures\_S & Onsite\_B & Onsite\_S \\ 
\hline
1. Sum($\mathcal{U}^s$) & 1436.2 & 1518.6 & 1365.3 & 2235.1 & 4563.7 & 5043.7 & 4285.1 & 6999.9 \\ 
2. Sum($\mathcal{U}^b$) & 141.8 & 82.6 & 183.7 & $-$103.3 & 491.5 & 189.3 & 546.7 & $-$349.2 \\ 
3. Sum($t^{Comp}$) & 236.4 & 294.1 & 266.8 & 266.8 & 614.3 & 889.7 & 825.6 & 825.6 \\ 
4. UFair & 0 & 0 & 0.05 & 0.14 & 0 & 0 & 0.0289 & 0.1072 \\ 
5. TFail & 0 & 0 & 1 & 1 & 0 & 0 & 1 & 1 \\ 
6. Sum(NC) & 26 & 26 & 14085 & 14085 & 42 & 42 & 32085 & 32085 \\ 
7. Sum(NL) & 0.286s & 0.286s & 154.935s & 154.935s & 0.462s & 0.462s & 0.462s & 352.935s \\ 
\hline
\end{tabular}}
\label{tab1}
}
\end{table*}

\begin{figure*}[h!t]
\centering
\subfigure[]{\includegraphics[width=.31\linewidth]{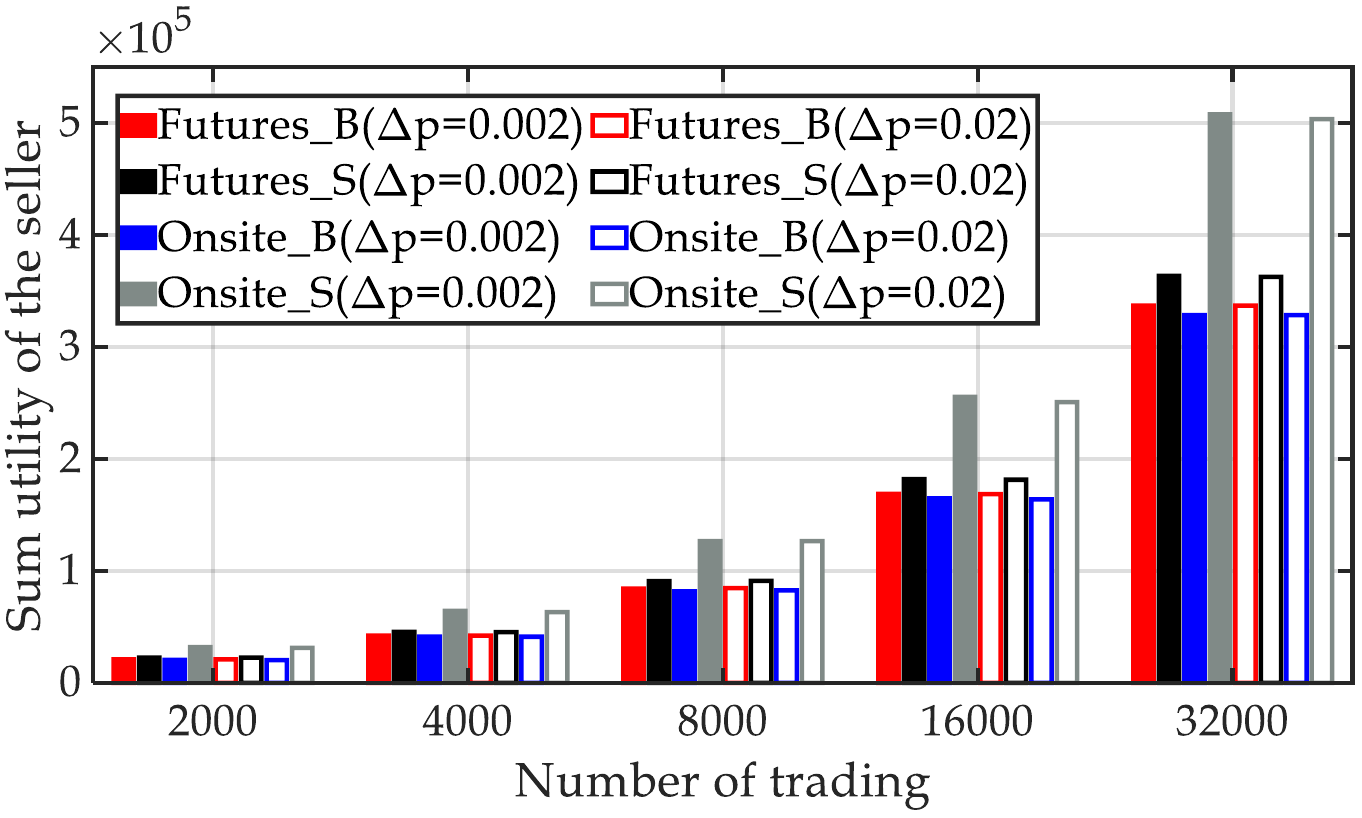}}\hfill
\subfigure[]{\includegraphics[width=.32\linewidth]{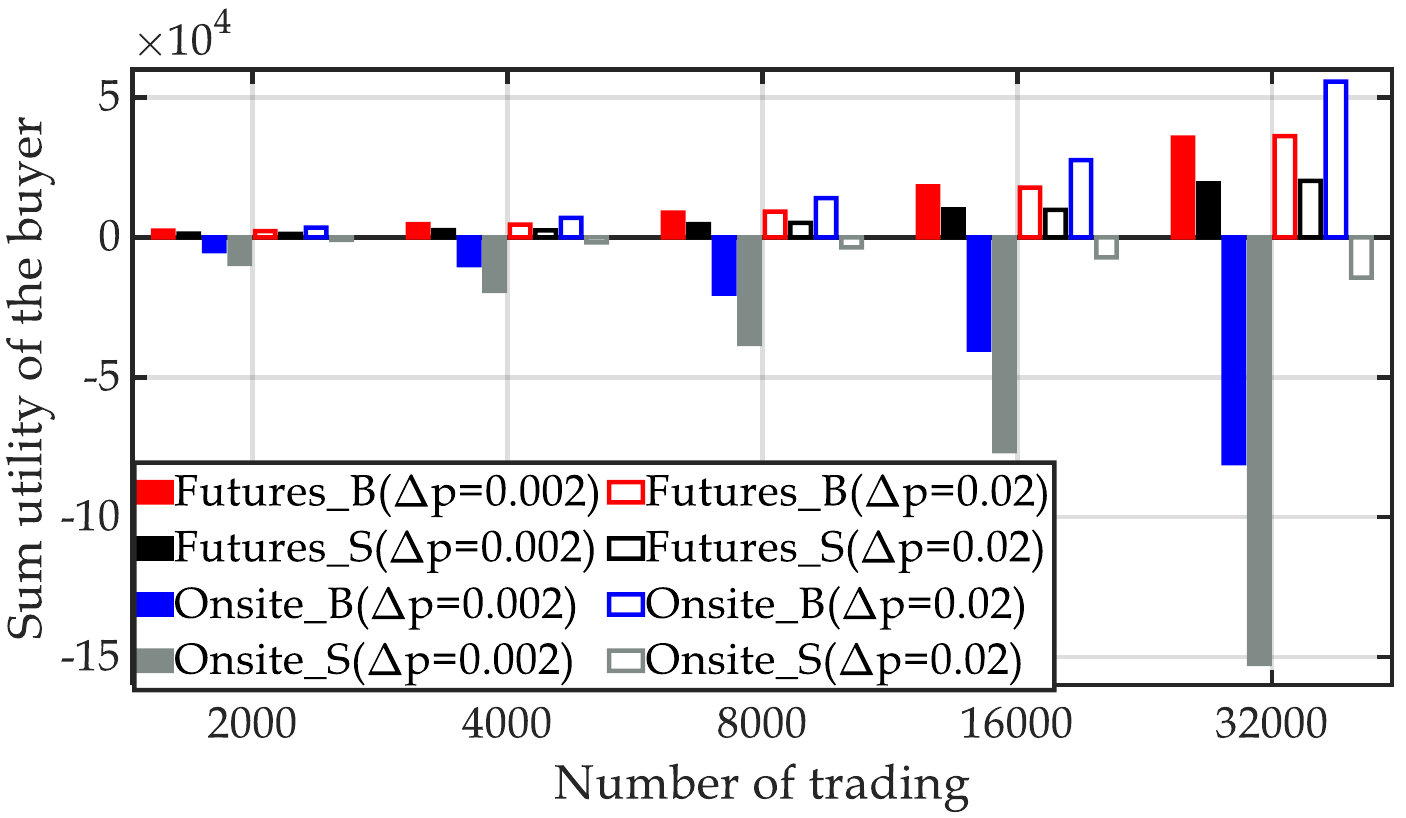}}\hfill
\subfigure[]{\includegraphics[width=.315\linewidth]{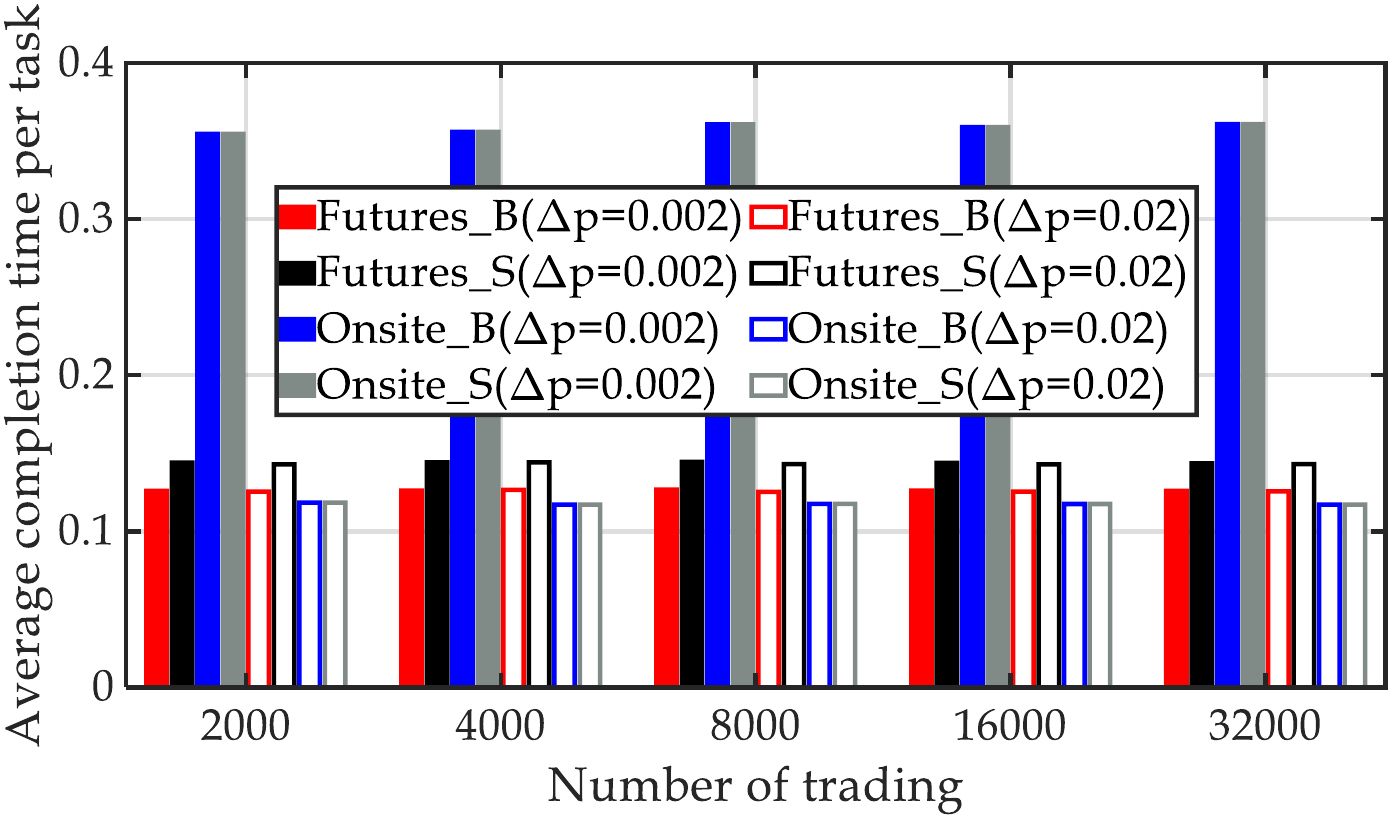}}\\
\vspace*{-.6em}
\subfigure[]{\includegraphics[width=.32\linewidth]{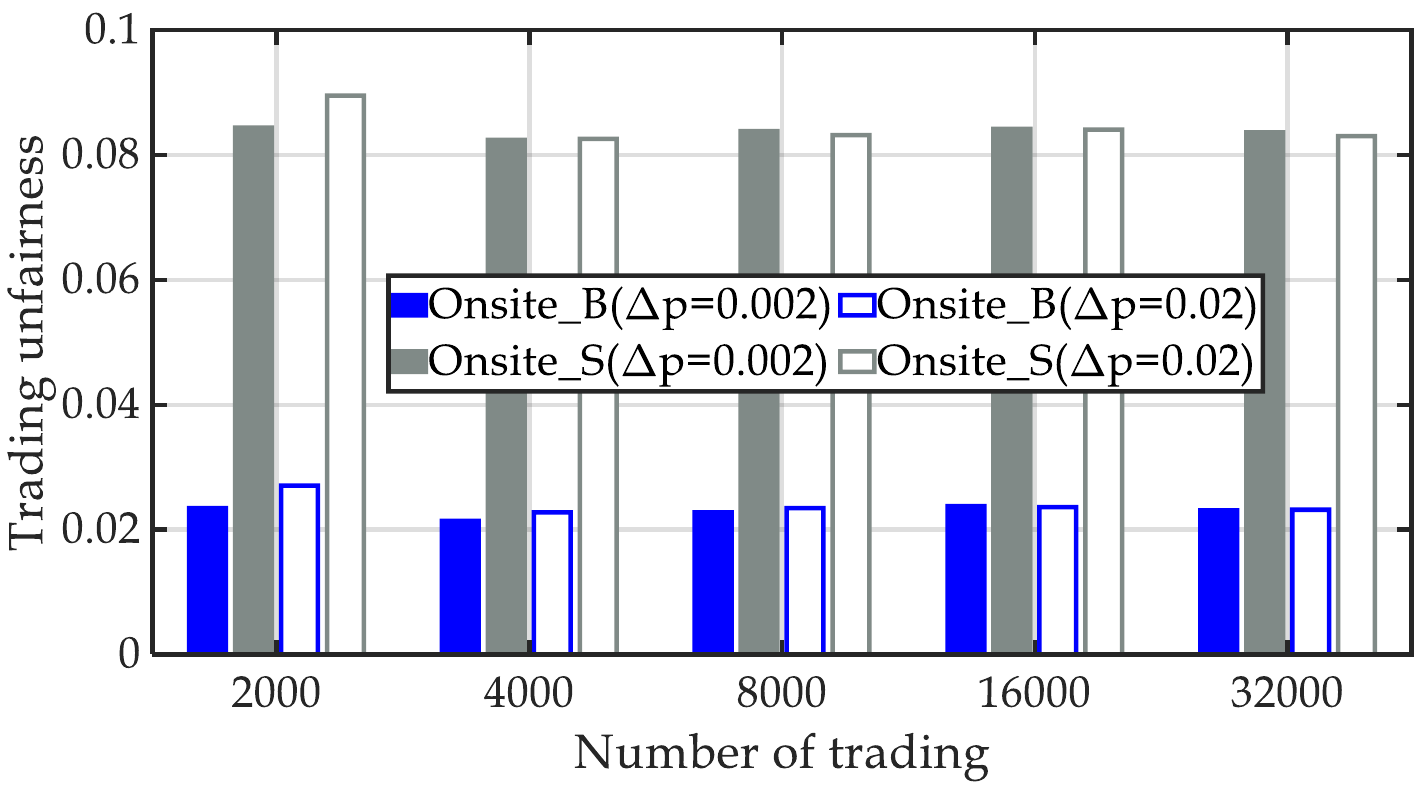}}\hfill
\subfigure[]{\includegraphics[width=.31\linewidth]{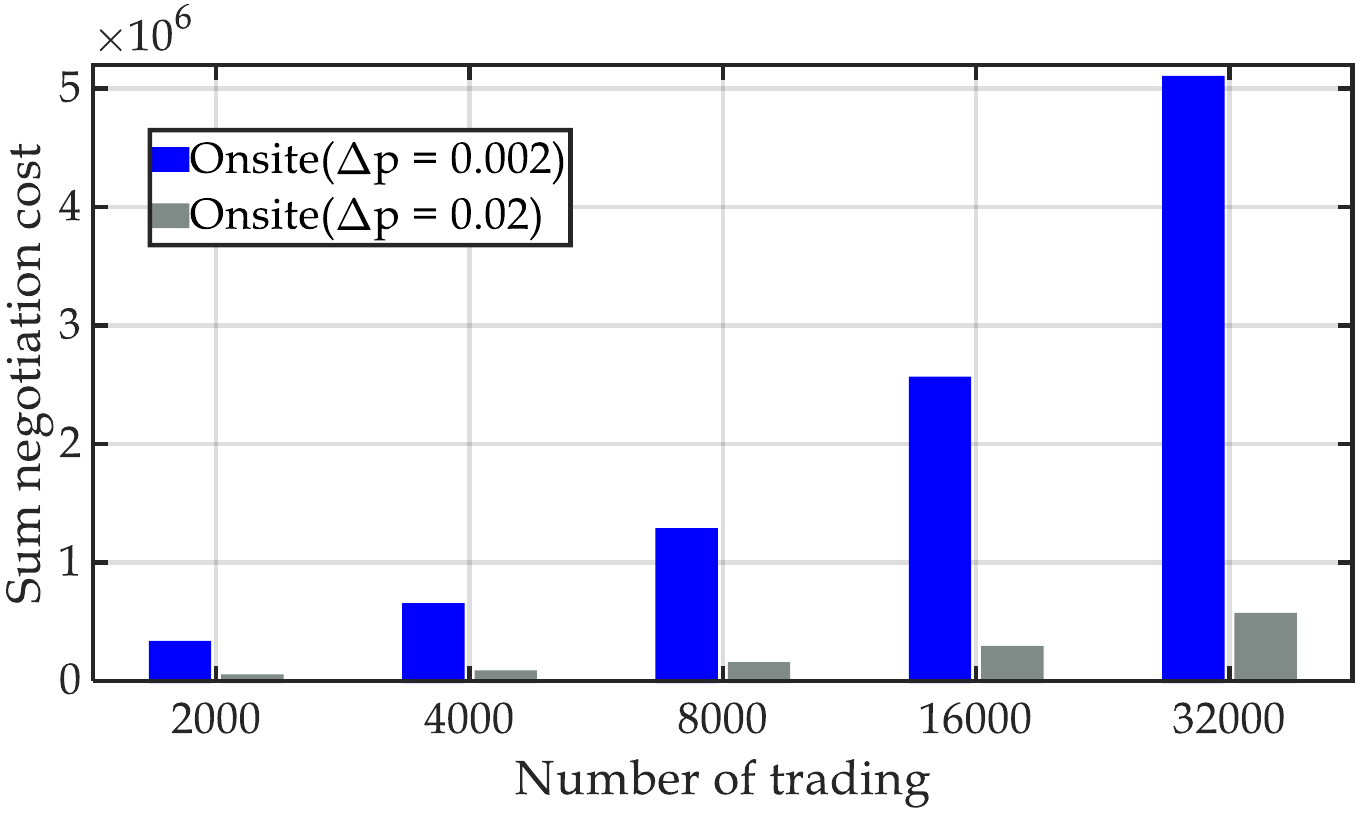}}\hfill
\subfigure[]{\includegraphics[width=.315\linewidth]{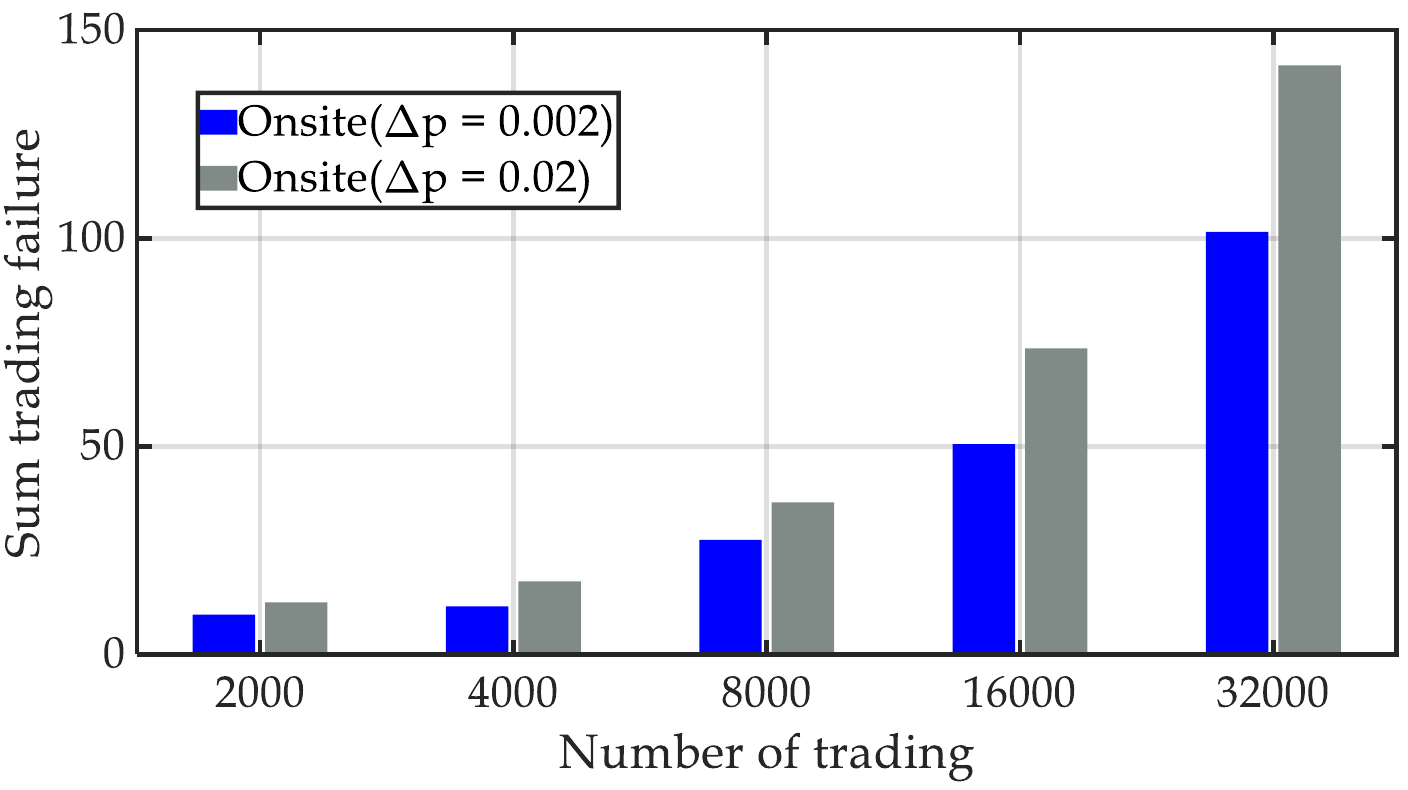}}
\vspace*{-.6em}
\caption{Performance comparison and evaluation of sum utility, negotiation costs, failures, average task completion time, and trading fairness.}
\label{fig5}
\end{figure*}

Futures\_B and Onsite\_B enhance the buyer's utility compared with Futures\_S and Onsite\_S since the buyer decides the final contract term by optimizing personal benefit. Notably, the proposed Futures\_B achieves better performance on the seller's and buyer's cumulative utility (see Sum($\mathcal{U}^b$) and Sum($\mathcal{U}^s$) in Table~\ref{tab1}) although Futures\_B may return unsatisfactory values in several trading due to unpredictable current resource supply, demand, and A2G channel quality. In Fig.~3(c), Fig.~3(i), and Table~\ref{tab1}, the proposed Futures\_B leads to swifter task completion per trading in most cases along with better sum task completion time (see Sum($t^{Comp}$)), Table~\ref{tab1}) compared with Futures\_B. Onsite\_S achieves slightly faster task completion (both per trading and cumulatively) than the proposed Futures\_S at the expense of the buyer's utility based on paying higher resource prices, thus leading to an unsatisfactory trading experience for the buyer. Additionally, in Figs.~3(d)-(f) and Figs.~3(j)-(l), the onsite-based mechanisms suffer from certain degrees of price fluctuations, heavy negotiation latency, and cost as well as failures, leading to general instability in the resource-trading system. In detail, onsite players may have to face heavy negotiation costs and long latency to reach trading terms. For example, the negotiation cost and latency spent by onsite players are respectively 541.7 and 763.9~times greater than those under the proposed futures-based mechanisms, which greatly reduces trading efficiency in wireless communication environments. Consequently, the proposed futures-based trading mechanisms can bring mutually beneficial utilities to both the players, while outperforming the onsite trading mechanisms with respect to unfairness, negotiation latency, cost, and failures.
\begin{figure}[h!t]
\centering
\subfigure[]{\includegraphics[width=.53\linewidth]{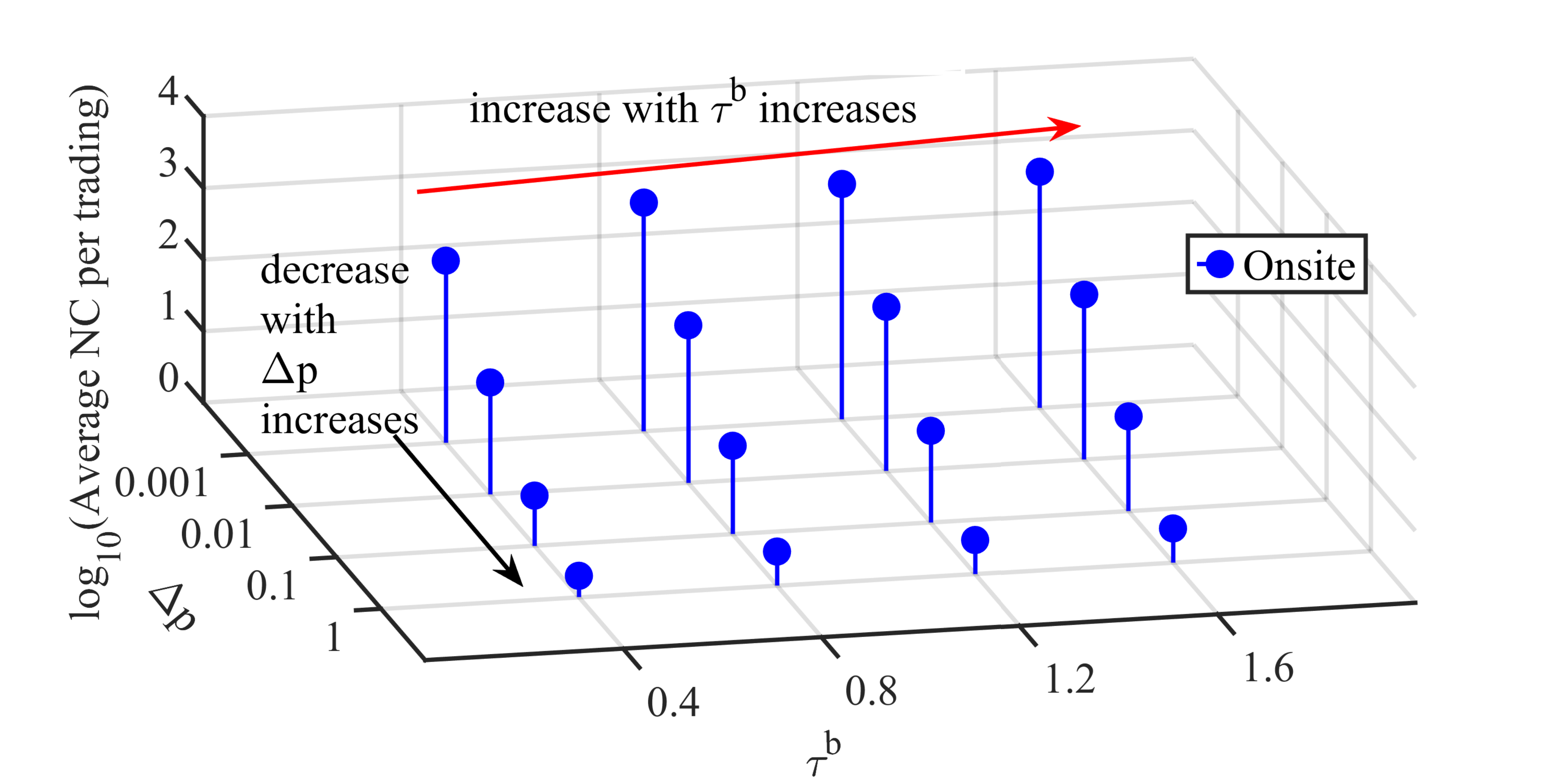}}~
\subfigure[]{\includegraphics[width=.53\linewidth]{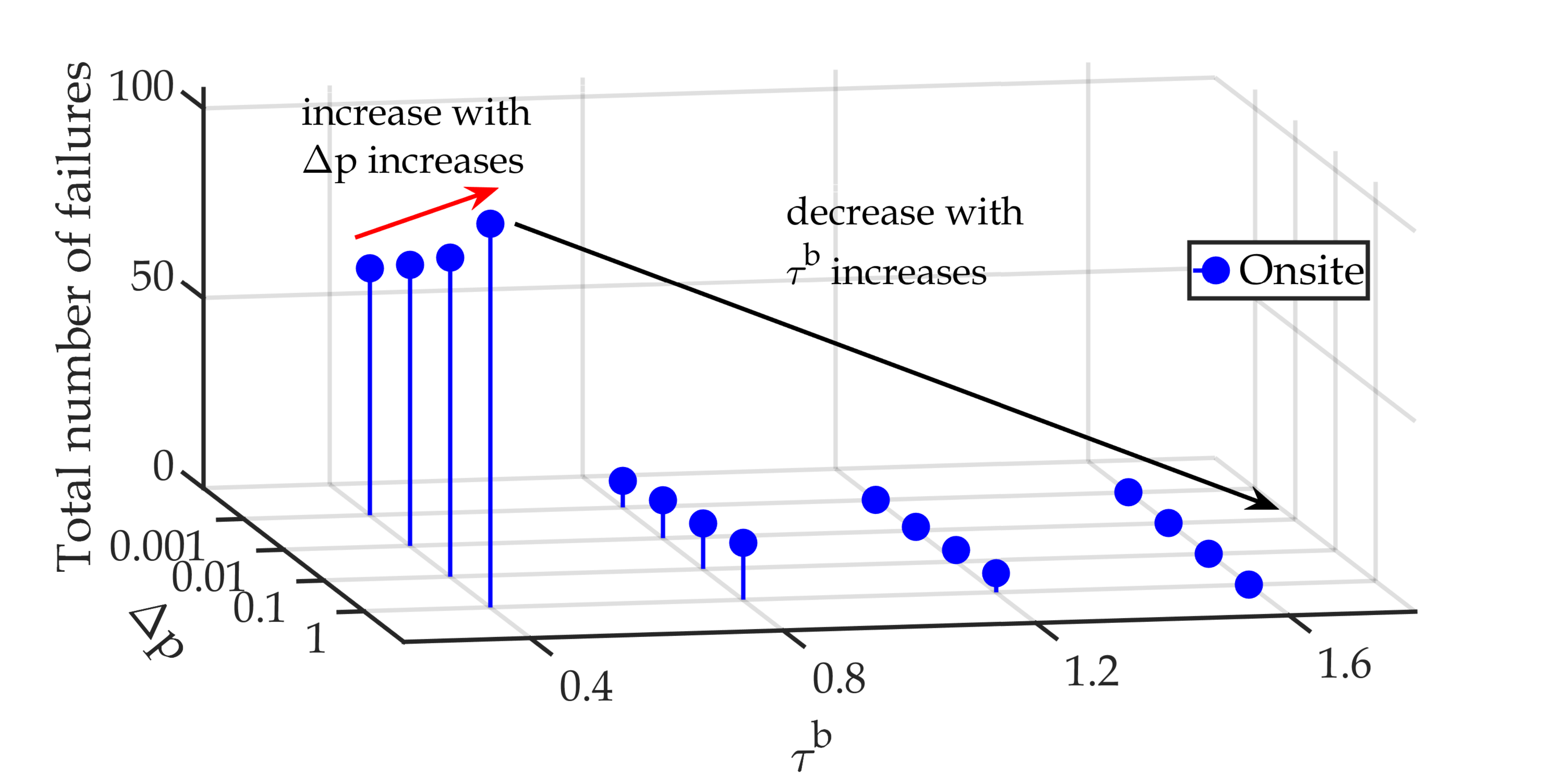}}
\vspace*{-.6em}
\caption{Effects on negotiation latency and trading failures upon having various $\Delta p$ and $\tau^b$.}
\label{fig6}
\vspace*{-.9em}
\end{figure}

Apparently, it is difficult to fully reflect the advantages of the proposed futures-based mechanism through a single trading, due to the unpredictability of network conditions, resource supply and demand. Fig. 4~thus demonstrates the performance evaluation and comparison of the players' cumulative utility (Figs.~4(a) and Fig.~4(b)), negotiation cost (Fig.~4(e)), and trading failures (Fig.~4(f)) via concerning a large number of trading as indicated by the Monte~Carlo~method. Fig.~4 also displays the associated average task completion time (Fig.~4(c)) and trading unfairness (Fig.~4(d)), comparing the proposed futures-based and onsite-based mechanisms by considering different values of $\Delta p$. Figs.~4(a)-(b) reveal that the proposed Futures\_B achieves slightly higher cumulative utility for the seller but suffers from lower such utility for the buyer compared with Onsite\_B when $\Delta p=0.02$. In this case, the negotiation latency and cost decline as the value of $\Delta p$ increases, which, however, will lead to more trading failures as shown in Fig.~4(f). Moreover, the seller of Onsite\_S always obtains greater cumulative utility compared with the proposed Futures\_S by sacrificing the buyer's benefits, which hinders the development of a mutually beneficial resource trading system. To alleviate unexpected failures by considering a smaller value of $\Delta p$ ($\Delta p=0.002 $), the proposed Futures\_B achieves better performance for both players' cumulative utility than Onsite\_B, while the proposed Futures\_S obtains a considerable cumulative utility for the seller (although slightly lower than that of Onsite\_S) and far better performance in terms of the buyer's cumulative utility. Related explanations appear in Fig.~4(e), where a smaller $\Delta p$ results in onsite-based mechanisms with drastically heavier negotiation costs and renders them unfeasible in wireless communication environments. 

Fig.~4(c) presents a comparison and evaluation of average task completion time considering different mechanisms. The proposed futures-based mechanisms consistently obtain far better task completion performance compared with onsite-based mechanisms when $\Delta p=0.002 $. In the case of $\Delta p=0.02$, although the onsite-based mechanisms achieve slightly better task completion performance, either the seller or the buyer will always sacrifice its benefit as depicted in Fig.~4(a) and Fig. 4(b). The trading unfairness of the two onsite-based methods is illustrated in Fig.~4(d), where Onsite\_S leads to worse fairness by choosing the final trading term with the largest available price, compared with Onsite\_B. 

\begin{figure*}[t!]
\centering
\subfigure[]{\includegraphics[width=.317\linewidth]{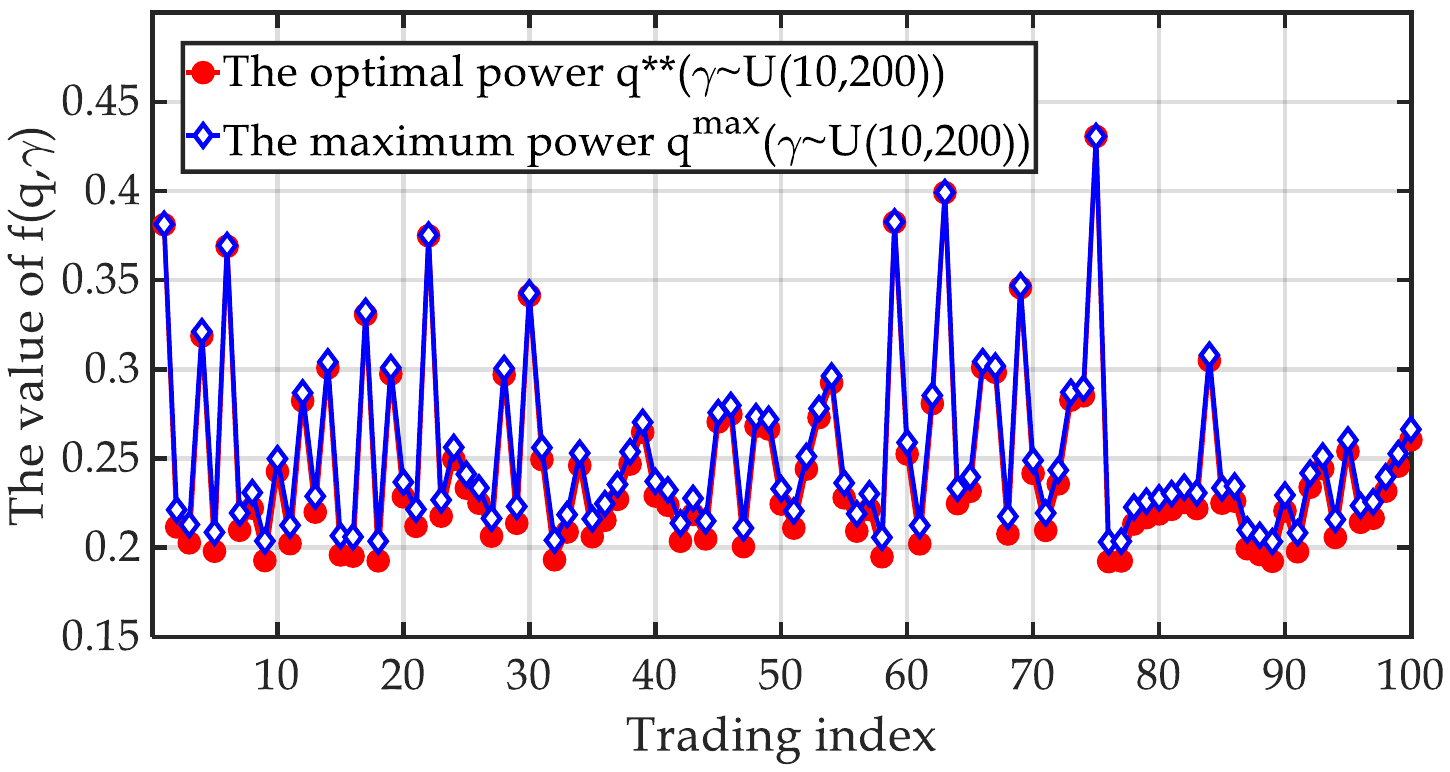}}\hfill
\subfigure[]{\includegraphics[width=.31\linewidth]{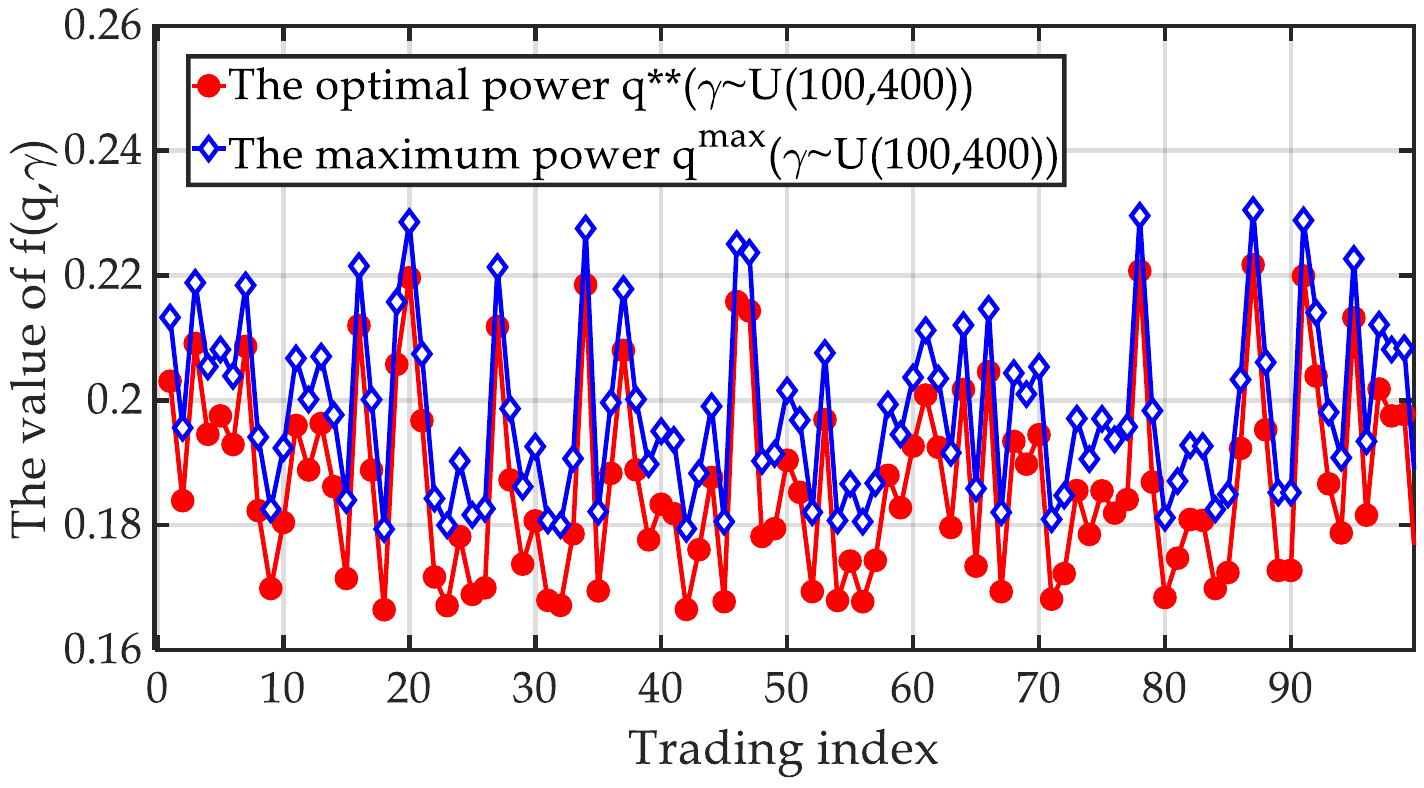}}\hfill
\subfigure[]{\includegraphics[width=.325\linewidth]{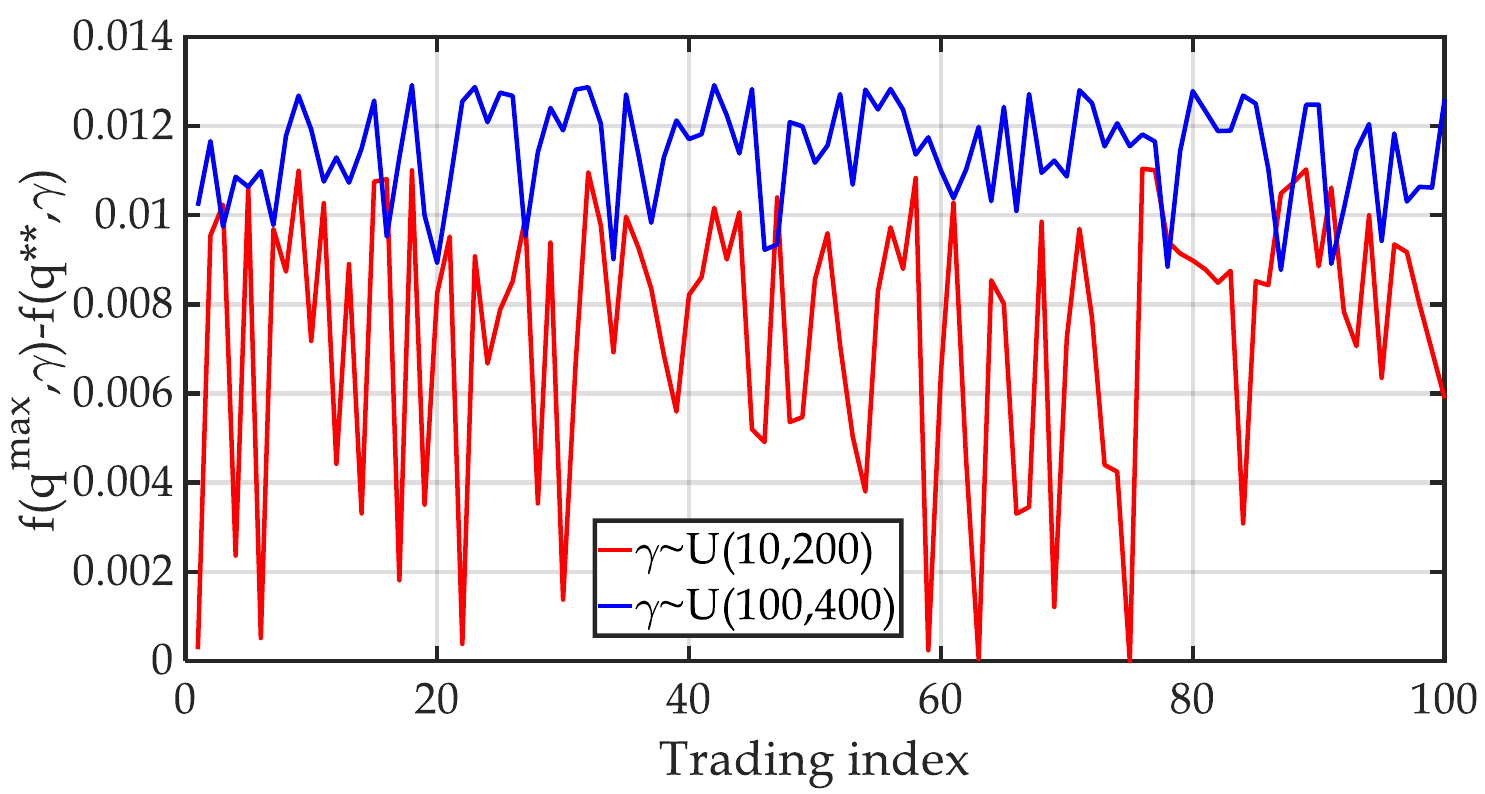}}
\vspace*{-.6em}
\caption{Performance on the value of $f(q, \gamma)$ via applying $q^{**}$ and $q^{max}$.}
\label{fig7}
\end{figure*}

Fig.~5 illustrates the performance of the onsite-based mechanism on terms of negotiation cost and trading failures, upon having various values of $\Delta p$ and $\tau ^b$ under 10{,}000~trading. For analytical simplicity, we apply a 10-based logarithm representation in Fig.~5(a) since the gap between the average negotiation cost (ANC) per trading may be excessively large under different $\Delta p$ and $\tau ^b$. As shown in Fig.~5(a), the ANC declines as the value of $\Delta p$ increases owing to a smaller number of quotation rounds. On the contrary, the rising value of $\tau^b$ leads to growth in ANC because the buyer enjoys faster task execution thanks to trading. In particular, a buyer with worse computational capability will be more willing to pay for computing service, which carries heavy onsite negotiation costs and latency according to the preceding evaluation. For example, on average, onsite players undergo 158 rounds and 2{,}000 rounds of quotation to reach a trading consensus when $\Delta p=0.01$, $\tau ^b=0.8$ and $\Delta p=0.001$, $\tau ^b=1.6$, respectively. Fig.~5(b) demonstrates the impact of varying $\Delta p$ and $\tau ^b$ on trading failures. Different from Fig.~5(a), the total number of trading failures grows as the value of $\Delta p$ increases; that is, the buyer may not be able to afford the price under current circumstances (e.g., a small value of $\gamma $), which leads to negative utility. Moreover, a rising $\tau^b$ will greatly reduce unexpected trading failures due to the buyer's restricted computational capability; in other words, buying computing services from the seller may bring the buyer greater benefits than local computing. As such, onsite players can achieve lower negotiation cost and latency with the decrease of $\Delta p$, which, however, will bring them undesired trading failures. Additionally, an onsite buyer with restricted computational capability will suffer from heavy negotiation latency and cost, which poses challenges to the power- and battery-constrained UAV. Onsite trading mechanisms face challenges to achieve a considerable trade-off between negotiation cost and failures, as well as mutually beneficial player utility (see Figs.~3-4), particularly in wireless communication environments with constrained mobile users.

The performance evaluation of the proposed power optimization algorithm is presented in Fig.~6, comparing the proposed optimal power $q^{**}$ and maximum power $q^{max}$ ($q^{max}=1000 $mW). The value of $f(q, \gamma)$ is considered to illustrate the advantages of the proposed algorithm, according to Section~IV-B. Fig.~6(a) compares the value of $f(q, \gamma)$ relative to $\gamma \sim {\rm U}(10{,}200)$ (where the relevant SNR reaches from roughly 10~dB to 23~dB). In this circumstance, the A2G communication link may be at risk of poor channel quality, which forces the buyer to raise the transmission power and thus leads to small gaps between the values of $f(q^{**}, \gamma)$ and $f(q^{max}, \gamma)$. Within a better communication environment as indicated by $\gamma \sim {\rm U}(100{,}400)$ (where the relevant SNR spans from roughly 20~dB to 26~dB) in Fig.~6(b), the distance between curves $f(q^{**}, \gamma)$ and $f(q^{max}, \gamma)$ expands because the buyer enjoys desirable channel quality in most trading. In particular, the proposed power optimization algorithm realizes greater buyer utility compared to applying $q^{max}$. 
Fig.~6(c) investigates the difference between $f(q^{**}, \gamma)$ and $f(q^{max}, \gamma)$ associated with Fig.~6(a) (see the red curve in Fig.~6(c)), and Fig.~6(b) (see the blue curve in Fig.~6(c)). These results support the advantage of the proposed power optimization algorithm upon having varying A2G channel qualities. 

\section{Conclusion}
To resolve the challenges of trading failures and unfairness, as well as the negotiation latency and cost associated with the onsite trading, we develop a fast futures-based resource trading mechanism under EC-UAV architecture. In this paper, a UAV (buyer) and an MEC server (seller) sign a mutually beneficial and risk-tolerable forward contract in advance to be fulfilled in the future. The mechanism addresses two key problems: first, the contract design is formulated as a MOO problem, for which we propose an efficient bilateral negotiation scheme to facilitate both players reaching a consensus on the amount of resources and the relevant price. Then, the power optimization problem is investigated to maximize the buyer's utility during each trading; specifically, a practical optimization algorithm is introduced via applying convex optimization techniques. Simulation results demonstrate that the proposed futures-based trading mechanism outperforms the onsite trading mechanism on significant indicators, while achieving mutually beneficial utility for the seller and the buyer. It is interesting to consider multiple buyers and more factors about  the unpredictable nature of the resource trading system, as well as smart forward contract, which will be investigated in our future work.

\appendices

\section*{Appendix}
\vspace{-0.2in}
\noindent
\subsection{Derivation of seller's expected utility $\overline{\mathcal{U}^s}(n_l, \mathcal{A}, \mathcal{P})$}

\noindent
Apparently, we have ${\rm E}[\mathcal{A}\mathcal{P}]=\mathcal{A}\mathcal{P}$ and ${\rm E}[U^s]=\frac{p_lM}{2}$, where ${\rm E}[\cdot]$ indicates the notation of expectation. As for ${\rm E}[C^s]$, we discuss the following cases according to the discrete values of $n_l$.

\noindent
$\bullet$ \textit{Case 1} $(\mathcal{A}=V=1 )$: the PMF of $C^s$ is shown in (26), based on which, we have ${\rm E}[C^s]=\frac{Mr_l}{M+1}$. 
\begin{align}
\label{eq26}
\Pr (C^s=k) =
\begin{cases}
{1}/{(M+1)}, & k=0 \vspace{-2ex} \\ 
{M}/{(M+1)}, & k=r_l 
\end{cases}
\end{align}

\noindent
$\bullet$ \textit{Case 2} $(\mathcal{A}=1, V>1)$: we have the PMF of $C^s$ as (27), based on which, ${\rm E}[C^s]$ is calculated by ${\rm E}[C^s]=\frac{Mr_l-Vr_l+r_l}{M+1}$.
\begin{align}
\label{eq27}
{\Pr (C^s=k)} =
\begin{cases}
{V}/{(M+1)}, & k=0 \vspace{-2ex} \\ 
{(M-V+1)}/{(M+1)}, & k=r_l 
\end{cases}
\end{align}

\noindent
$\bullet$ \textit{Case 3} ($1<\mathcal{A}=V$): the PMF of $C^s$ is given by (28), and we have ${\rm E}[C^s]=\frac{-V^2r_l+2VMr_l+Vr_l}{2(M+1)}$.
\begin{align}
\label{eq28}
& {\Pr (C^s=k)}=
\begin{cases}
{1}/{(M+1)}, & k=0 \vspace{-2ex} \\ 
{1}/{(M+1)}, & k\in \{r_l, 2r_l, \cdots, (V-1) r_l\} \vspace{-2ex}  \\
{(M-V+1)}/{(M+1)}, & k=Vr_l 
\end{cases} 
\end{align}

\noindent
$\bullet$ \textit{Case 4} $(1<\mathcal{A}<V)$: the PMF of $C^s$ is shown below, 
\begin{align}
\label{eq29}
& {\Pr (C^s=k)} =
\begin{cases}
{(V-\mathcal{A}+1)}/{(M+1)}, & k=0 \vspace{-2ex} \\ 
{1}/{(M+1)}, & k\in \{r_l, {2r}_l, (\mathcal{A}-1) r_l\} \vspace{-2ex} \\ 
{(M-V+1)}/{(M+1)}, & k=\mathcal{A}r_l
\end{cases}. 
\end{align}

Thus, the calculation of ${\rm E}[C^s]$ is given by (30):
\begin{align}
\label{eq30}
& {\rm E}[C^s]=\frac{V-\mathcal{A}+1}{M+1}\times 0+\sum^{k=(\mathcal{A}-1) r_l}_{k=r_l}{\frac{k}{M+1}}+\frac{r_l\mathcal{A}(M-V+1)}{M+1}\\ \notag
&=\frac{r_l\mathcal{A}^2}{2(M+1)}+\frac{(r_l+2Mr_l-2Vr_l)\mathcal{A} }{2(M+1)}.
\end{align}

Apparently, (30) works for any $\mathcal{A} \in \{1,2,...,V\}$, and $\overline{\mathcal{U}^s}(n_l, \mathcal{A}, \mathcal{P})$ can thus be calculated by (3).

\subsection{Derivation of seller's risk $\mathcal{R}^s(n_l, \mathcal{A}, \mathcal{P})$}

\noindent
According to (6), let random variables $S_1=n_lp_l$, $S_2=n_lp_l-{n_lr}_l+r_l(V-\mathcal{A})$, and $S_3=n_lp_l-r_l\mathcal{A}$, we compute the PMF of $S$ as given in (31).
\begin{align}
\label{eq31}
& \Pr (S=k)= 
\begin{small}
\begin{cases}
\Pr (S_1=k')={1}/{(M+1)},~ k'\in \{0, p_l, 2p_l, \cdots, (V-\mathcal{A}) p_l\} \vspace{-2ex} \\ 
\Pr (S_2=k'')={1}/{(M+1)},~k''\in \{(V-\mathcal{A}) p_l+p_l-r_l,\cdots, (V-\mathcal{A}) p_l+\mathcal{A}(p_l-r_l)\} \vspace{-2ex} \\
\Pr (S_3=k''')={1}/{(M+1)},~k'''\in \{(V+1) p_l-r_l\mathcal{A}, \cdots, Mp_l-r_l\mathcal{A}\} 
\end{cases}
\end{small}\notag \\
& ={1}/{(M+1)}, k\in \{0, p_l, \cdots, (V-\mathcal{A}) p_l, (V-\mathcal{A}) p_l+p_l-r_l, \cdots, (V-\mathcal{A}) p_l+\mathcal{A}(p_l-r_l),\notag \vspace{-2ex} \\
& \cdots, (V+1) p_l-r_l\mathcal{A}, \cdots, Mp_l-r_l\mathcal{A}\}
\end{align}

Correspondingly, the risk of the seller is calculated by (7) based on the CDF of $S$.
\subsection{Derivation of buyer's expected utility $\overline{\mathcal{U}^b}(q, \gamma, n_b,$ $\mathcal{A}, \mathcal{P})$}

\noindent
Let random variable $X={[\mathcal{A}, n_b]}^-$, we consider the following two cases.

\noindent
$\bullet$ \textit{Case 1} ($\mathcal{A}>1 $): we discuss the following conditions. When $x<1$, we have CDF ${\rm F}_X(x)=\Pr (X\le x)=0 $; when $x>\mathcal{A}$, we have ${\rm F}_X(x)=1$, and $\Pr (X\le x)=0 $; when $x=\mathcal{A}$, we have ${\rm F}_X(x)=1$, and $\Pr (X\le x)=\frac{N-\mathcal{A}+1}{N}$; and when $1\le x<\mathcal{A}$, we have $\Pr (X=x)=\frac{1}{N}$, and the relevant ${\rm F}_X(x)$ is calculated as (32).
\begin{align}
\label{eq32}
{\rm F}_X(x)& =1-{\Pr (X>x)}=1-\Pr (n_b>x)=\Pr (n_b\le x)= F_{n_b}(x)=\frac{x}{N} 
\end{align}
Consequently, we have the CDF of random variable $X$ as:
\begin{align}
\label{eq33}
{\rm F}_X(x)=
\begin{cases}
0, & x<1 \vspace{-1.2ex} \\ 
\dfrac{x}{N}, & x\in \{1, \cdots, \mathcal{A}-1\} \vspace{-1.2ex} \\ 
\vspace{-0.05in}
1, & x\ge \mathcal{A} 
\end{cases}, 
\end{align}
The PMF of $X$ is calculated below when $1< \mathcal{A}\le V$: 
\begin{align}
\label{eq34}
\Pr (X=x)=
\begin{cases}
\dfrac{1}{N}, & x\in \{1, \cdots, \mathcal{A}-1\} \vspace{-1.2ex} \\[2.5mm] 
\dfrac{N-\mathcal{A}+1}{N}, & x=\mathcal{A} \vspace{-1.2ex} \\ 
0, & {\text{otherwise}} 
\end{cases}. 
\end{align}
\noindent
$\bullet$ \textit{Case 2} $(\mathcal{A}=1)$: in this case, we have $ X=1$, and the CDF and PMF of $X$ are thus calculated as (35) and (36):
\begin{align}
\label{eq35}
& {\rm F}_X(x)=
\begin{cases}
0, & x<1 \vspace{-2ex} \\ 
1, & x\ge 1 
\end{cases}, \\
\label{eq36}
& \Pr (X=x)=
\begin{cases}
1, & x=1 \vspace{-2ex} \\ 
0, & \text{otherwise} 
\end{cases}. 
\end{align}
Correspondingly, the expected value of $X$ for any $1\le \mathcal{A}\le V$ is calculated by (37).
\begin{align}
\label{eq37}
{\rm E}[X] & =
\begin{cases}
1\times 1, & \mathcal{A}=1 \vspace{-1ex} \\ 
\displaystyle\sum^{x=\mathcal{A}-1}_{x=1}{\dfrac{x}{N}}+\mathcal{A}\times \dfrac{N-\mathcal{A}+1}{N}, & 1< \mathcal{A}\le V 
\end{cases}
=\frac{-\mathcal{A}^2\!+\! (2N\!+\! 1)\mathcal{A}}{2N} 
\end{align}
Let random variable $ Y=\frac{1}{\log _2(1+q\gamma)}$, we discuss the CDF and PDF of $Y$ based on the distribution of $\gamma$, given by (38) and (39), respectively.
\begin{align}
\label{eq38}
& {\rm F}_Y(y)=\Pr \left(\gamma \ge \frac{{2}^{\frac{1}{y}}-1}{q}\right)=
\begin{cases}
0, & y<\dfrac{1}{\log _{2}({q\varepsilon}_{2}+1)} \vspace{-0.8ex} \\ 
1-\dfrac{{2}^{\frac{1}{y}}\!-\! q{\varepsilon}_{1}\!-\! 1}{{q\varepsilon}_{2}\!-\! {q\varepsilon}_{1}}, & \dfrac{1}{\log _{2}({q\varepsilon}_{2}\!+\! 1)}\le y\le \dfrac{1}{\log _{2}({q\varepsilon}_{1}\!+\! 1)} \vspace{-0.6ex} \\ 
1, & y>\dfrac{1}{\log _{2}({q\varepsilon}_{1}+1)} 
\end{cases}\\
\label{eq39}
& {\Pr (Y=y)} =\frac{\partial {\rm F}_Y(y)}{\partial y}
= \begin{cases}
\dfrac{\ln 2}{q{\varepsilon}_{2}-q{\varepsilon}_{1}}\times \dfrac{{2}^{\frac{1}{y}}}{y^{2}}, & y\in \left[\dfrac{1}{\log _{2}(1\!+\! q{\varepsilon}_{2})}, \dfrac{1}{\log _{2}(1\!+\! q{\varepsilon}_{1})}\right] \vspace{-1ex} \\ 
0,  & \text{otherwise} 
\end{cases}
\end{align}

Correspondingly, the expectation of random variable $Y$ is calculated as:
\begin{align}
\label{eq40}
& {\rm E}[Y]={\rm E}\left[\frac{1}{\log _2(1+q\gamma)}\right]=\int^{\frac{1}{\log _2(1+q\varepsilon _1)}}_{\frac{1}{\log _2(1+q\varepsilon _2)}}{y{\Pr (Y=y)} {\rm d}y}
= \frac{\ln 2}{q\varepsilon _2-q\varepsilon _1}\int^{\frac{1}{\log _2(1+q\varepsilon _1)}}_{\frac{1}{\log _2(1+q\varepsilon _2)}}{\left(\frac{2^{\frac{1}{y}}}{y}\right) {\rm d}y} \notag \\
& =\frac{\ln 2}{q\varepsilon _2-q\varepsilon _1}({\rm Ei}(\ln 2\times \log _2(1+q\varepsilon _2))-{\rm Ei}(\ln 2\times \log _2(1+q\varepsilon _1 ))).
\end{align}

Notably, ${\rm Ei}(\cdot)$ indicates the exponential integral function defined by ${\rm Ei}(y)=\int^y_{-\infty} {\frac{e^x}{x}}{\rm d}x$. Let $y_1=\ln 2\times \log _2(q\varepsilon _1+1)$ and $y_2=\ln 2\times \log _2(q\varepsilon _2+1)$ for notational simplicity, (40) is further considered as (41).
\begin{align}
\label{eq41}
{\rm E}[Y]\!=\! \frac{\ln 2\times \left(\int^{y_2}_{-\infty} {\frac{e^x}{x}}{\rm d}x-\int^{y_1}_{-\infty} {\frac{e^x}{x}}{\rm d}x\right)}{q\varepsilon _2-q\varepsilon _1}\!=\! \frac{\ln 2\times \int^{y_2}_{y_1}{\frac{e^x}{x}}{\rm d}x}{q\varepsilon _2-q\varepsilon _1} 
\end{align}

Correspondingly, because the random variables $X$ and $Y$ are independent~of each other, $\overline{\mathcal{U}^b}$ is given by (42):
\begin{align}
\label{eq42}
& \overline{\mathcal{U}^b}(q, \gamma, n_b, \mathcal{A}, \mathcal{P})={\rm E} \left[X\tau ^b-\frac{XD}{W\log _2(1+q\gamma)}\right]-\omega _2{\rm E}\left[\frac{qDX}{W\log _2(1+q\gamma)}\right]-\tau ^s-\omega _1\mathcal{A}\mathcal{P}-\omega _2\ell  \notag  \\ 
&={\rm E}[X] \left(\tau ^b-\frac{D+\omega _2qD}{W}\times {\rm E}[Y]\right)-\tau ^s-\omega _1\mathcal{A}\mathcal{P}-\omega _2\ell. 
\end{align}

Let $\mathbb{C}_2=\tau ^b-\frac{D+\omega _2qD}{W}\times {\rm E}[Y]$ denote a constant under any given $q$ for notational simplicity, $\overline{\mathcal{U}^b}(q, \gamma, n_b, \mathcal{A}, \mathcal{P})$ can be thus calculated as (11).
\subsection{Derivation of buyer's risk $\mathcal{R}^b(q, \gamma, n_b, \mathcal{A}, \mathcal{P})$}
\noindent
Let random variable $X={[\mathcal{A}, n_b]}^-$, $Y=\frac{1}{\log _{2}(1+q\gamma)}$, and $Z=\tau ^b-\frac{D+{\omega}_{2}qD}{W}Y$. Apparently, $X$ and $Z$ are independent of each other. We first discuss the CDF of $Z$ based on (38), as given in (43).
\begin{align}
\label{eq43}
{\rm F}_Z(z) & =1-\Pr \left(Y\le \frac{W(\tau ^b-z)}{D+{\omega}_{2}qD}\right)=
\begin{cases}
0, & z<\mathbb{C}_{4} \vspace{-1.2ex} \\ 
\dfrac{{2}^{\frac{D+{\omega}_{2}qD}{W(\tau ^b-z)}}-q{\varepsilon}_{1}-1}{{q\varepsilon}_{2}-{q\varepsilon}_{1}}, & \mathbb{C}_{4}\le z\le \mathbb{C}'_{4}, \vspace{-1.2ex} \\ 
1, & z> \mathbb{C}'_{4} 
\end{cases}
\end{align}
where $\mathbb{C}_{4}=\tau ^b-\frac{D+{\omega}_{2}qD}{W\log _{2}({q\varepsilon}_{1}+1)}$ and $\mathbb{C}'_{4}=\tau ^b-\frac{D+{\omega}_{2}qD}{W\log _{2}({q\varepsilon}_{2}+1)}$ for notational simplicity. Owing to the~discrete values of $n_b $ and $\mathcal{A}$, the calculation of $\mathcal{R}^b(q, \gamma, n_b, \mathcal{A}, \mathcal{P})$ is discussed via considering the following two cases.

\noindent
$\bullet$ \textit{Case 1} $(\mathcal{A}=1 )$: consider $\mathcal{A}=1$, $\mathcal{R}^b(q, \gamma, n_b, \mathcal{A}=1, \mathcal{P})$ is given by (44).
\begin{align}
\label{eq44}
& \mathcal{R}^b(q, \gamma, n_b, \mathcal{A}=1, \mathcal{P})=\Pr \{Z\le \mathbb{C}'_{3}\}
= \begin{cases}
0, & \mathbb{C}'_{3}<\mathbb{C}_{4} \vspace{-1.2ex} \\ 
\dfrac{{2}^{\frac{D+{\omega}_{2}qD}{W(\tau ^b-\mathbb{C}'_{3})}}-q{\varepsilon}_{1}-1}{{q\varepsilon}_{2}-{q\varepsilon}_{1}}, & \mathbb{C}_{4}\le \mathbb{C}'_{3}\le \mathbb{C}'_{4} \vspace{-1.2ex} \\ 
1, & \mathbb{C}'_{3}>\mathbb{C}'_{4} 
\end{cases}
\end{align}

\noindent
$\bullet$ \textit{Case 2} $(\mathcal{A}>1 )$: consider $\mathcal{A}>1$, we have $\mathbb{C}_{3}=XZ$. Correspondingly, by analyzing the distribution of the product of discrete and continuous random variables, we calculate $\mathcal{R}^b(q, \gamma, n_b, \mathcal{A}>1, \mathcal{P})$ as (45) based on (33), (34), and (43). 
\begin{align}
\label{eq45}
& \mathcal{R}^b(q, \gamma, n_b, \mathcal{A}>1, \mathcal{P})=\Pr \{\mathbb{C}_{3}\le \mathbb{C}'_{3}\}=\Pr \{XZ\le \mathbb{C}'_{3}\} =\sum^{x=\mathcal{A}}_{x=1}{\Pr (X=x)\times {\rm F}_Z\left(\frac{\mathbb{C}'_{3}}{x}\right)} \notag \\
& =\frac{1}{N}\times \sum^{x=\mathcal{A}-1}_{x=1}{{\rm F}_Z\left(\frac{\mathbb{C}'_{3}}{x}\right)}+\frac{N-\mathcal{A}+1}{N}\times {\rm F}_Z\left(\frac{\mathbb{C}'_{3}}{\mathcal{A}}\right) 
\end{align}
%
\subsection{Derivation of the optimal power $q^{**}$}

\noindent
The first-order derivative of $h(\beta, \gamma)$ is calculated as (46):
\begin{align}
\label{eq47}
\frac{\partial h(\beta, \gamma)}{\partial \beta} =1+\frac{\omega _2(\beta \times 2^{\frac{1}{\beta}}-\beta -\ln 2\times 2^{\frac{1}{\beta}})}{\gamma \beta}. 
\end{align}

Correspondingly, we have (47) by letting $\frac{\partial h(\beta, \gamma)}{\partial \beta} =0$, $\beta '=\frac{1}{\beta}$ and $\beta''=\beta '\times \ln 2-1$. 
\begin{align}
\label{eq47}
& 2^{\frac{1}{\beta}}-\frac{2^{\frac{1}{\beta}}\ln 2}{\beta} =\frac{\omega _2-\gamma} {\omega _2}\Longrightarrow 
2^{\beta '}-\beta'2^{\beta'}\ln 2=\frac{\omega _2-\gamma} {\omega _2}\Longrightarrow \rm{e}^{(\beta '\ln 2)}\times (1-\beta'\ln 2)=\frac{\omega _2-\gamma} {\omega _2}\notag\\
&\xLongrightarrow{\beta''=\beta'\ln 2-1}-\rm{e}^{(\beta''+1)}\times \beta''=\frac{\omega _2-\gamma} {\omega _2} \Longrightarrow \rm{e}^{\beta''}\times \beta''=\frac{\gamma -\omega _2}{\rm{e}\times \omega _2} \Longrightarrow \beta''=\textbf{W}\left(\frac{\gamma -\omega _2}{\rm{e}\times \omega _2}\right),
\end{align}

\noindent
where $\textbf{W}(\cdot)$ denotes the Lambert W Function~\cite{43}. Thus, we have ${\beta}^*$ shown by (48).
\begin{align}
\label{eq48}
& {\beta}^*= 
\begin{cases}
\dfrac{\ln 2}{\textbf{W}\left(\frac{\gamma -\omega _2}{\rm{e}\times \omega _2}\right) \!+\! 1}, &\text{if } \dfrac{\ln 2}{\textbf{W}\left(\frac{\gamma -\omega _2}{\rm{e}\times \omega _2}\right) \!+\! 1}>\dfrac{1}{\log _2(1 \!+\! \gamma q^{max})} \vspace{-0.5ex} \\[5mm]
\dfrac{1}{\log _2(1+\gamma q^{max})}, & \text{otherwise } 
\end{cases}
\end{align}

Based on (48), the optimal transmission power is thus obtained as given in (24).

\vfill

\begin{thebibliography}{99}

\bibitem{1} T. Bai, J. Wang, Y. Ren, and L. Hanzo, ``Energy-Efficient Computation Offloading for Secure UAV-Edge-Computing Systems,'' \textit{IEEE Trans. Veh. Technol}., vol. 68, no. 6, pp. 6074--6087, 2019. 

\bibitem{2} Q. Wu, Y. Zeng, and R. Zhang, "Joint Trajectory and Communication Design for Multi-UAV Enabled Wireless Networks,"\textit{IEEE Trans. Wireless Commun.}, vol. 17, no. 3, pp. 2109-2121, 2018. 

\bibitem{3} N. Cheng, F. Lyu, W. Quan, C. Zhou, H. He, W. Shi, and X. Shen, ``Space/Aerial-Assisted Computing Offloading for IoT Applications: A Learning-Based Approach,'' \textit{IEEE J. Sel. Areas Commun}., vol. 37, no. 5, pp. 1117--1129, 2019. 

\bibitem{4} Y. Zeng, J. Xu and, R. Zhang, "Energy Minimization for Wireless Communication With Rotary-Wing UAV," \textit{IEEE Trans. Wireless Commun.}, vol. 18, no. 4, pp. 2329-2345, 2019.

\bibitem{5} L. Abend, ``Pilot: The mystery of Kobe Bryant's chopper crash,'' https://www. cnn. com/2020/01/28/opinions/kobebryant-helicopter-crash-abend/index. html, Jan. 28. 2020. 

\bibitem{6} V. Chamola, V. Hassija, V. Gupta, and M. Guizani, ``A Comprehensive Review of the COVID-19 Pandemic and the Role of IoT, Drones, AI, Blockchain, and 5G in Managing its Impact," \textit{IEEE Access}, vol. 8, pp. 90225-90265, 2020.

\bibitem{7} Y. Zeng, Q. Wu, and R. Zhang, "Accessing From the Sky: A Tutorial on UAV Communications for 5G and Beyond," \textit{Proc. IEEE}, vol. 107, no. 12, pp. 2327-2375, 2019. 

\bibitem{8} C. Liu, M. Bennis, M. Debbah, and H. V. Poor, ``Dynamic Task Offloading and Resource Allocation for Ultra-Reliable Low-Latency Edge Computing,'' \textit{IEEE Trans. Commun.}, vol. 67, no. 6, pp. 4132-- 4150, 2019. 

\bibitem{9} \"{O}. U. Akg\"{u}l, I. Malanchini, and A. Capone, ``Dynamic Resource Trading in Sliced Mobile Networks,'' \textit{IEEE Trans. Netw. Serv. Manag}., vol. 16, no. 1, pp. 220--233, 2019. 

\bibitem{10} L. H. Ederington, ``The Hedging Performance of The New Futures Markets,'' \textit{J. Finance}, vol. 34, no. 1, pp. 157--170, 1979. 

\bibitem{11} J. C. Hull, ``Options, Futures, and Other Derivatives,'' \textit{Pearson Education India}, 2003. 

\bibitem{12} I. Bajaj, Y. H. Lee, and Y. Gong, ``A Spectrum Trading Scheme for Licensed User Incentives,'' \textit{IEEE Trans. Commun}., vol. 63, no. 11, pp. 4026--4036, 2015. 

\bibitem{13} H. Guo, and J. Liu, ``UAV-Enhanced Intelligent Offloading for Internet of Things at the Edge,'' \textit{IEEE Trans. Ind. Informat}., vol. 16, no. 4, pp. 2737--2746, 2020. 

\bibitem{14} J. Xiong, H. Guo, and J. Liu, ``Task Offloading in UAV-Aided Edge Computing: Bit Allocation and Trajectory Optimization,'' \textit{IEEE Commun. Lett}., vol. 23, no. 3, pp. 538--541, 2019. 

\bibitem{15} X. Hu, K. Wong, K. Yang, and Z. Zheng, ``UAV-Assisted Relaying and Edge Computing: Scheduling and Trajectory Optimization,'' \textit{IEEE Trans. Wireless Commun}., vol. 18, no. 10, pp. 4738--4752, 2019. 

\bibitem{16} X. Hu, K. Wong, K. Yang, and Z. Zheng, ``Task and Bandwidth Allocation for UAV-Assisted Mobile Edge Computing with Trajectory Design." \textit{IEEE Int. Conf. Global. Commun. (GLOBECOM)}, Waikoloa, HI, USA, Dec. 2019, pp. 1--6. 

\bibitem{17} M. A. Messous, S. M. Senouci, H. Sedjelmaci, and S. Cherkaoui, ``A Game Theory based Efficient Computation Offloading in an UAV Network,'' \textit{IEEE Trans. Veh. Technol}., vol. 68, no. 5, pp. 4964--4974, 2019. 

\bibitem{18} M. A. Messous, A. Arfaoui, A. Alioua, and S. M. Senouci, ``A Sequential Game Approach for Computation-Offloading in an UAV Network,'' \textit{IEEE Int. Conf. Global. Commun. (GLOBECOM)}, Singapore, Dec. 2017, pp. 1--6. 

\bibitem{19} M. Liwang, Z. Gao, and X. Wang, ``Energy-aware Graph Job Allocation in Software Defined Air-Ground Integrated Vehicular Networks,'' arXiv preprint arXiv:2008.01144, 2020. 

\bibitem{20} M. LiWang, J. Wang, Z. Gao, X. Du, and M. Guizani, ``Game Theory based Opportunistic Computation Offloading in Cloud-Enabled IoV,'' \textit{IEEE Access}, vol. 7, pp. 32551--3256, 2019. 

\bibitem{21} Y. Wang, P. Lang, D. Tian, J. Zhou, X. Duan, Y. Cao, and D. Zhao, ``A Game-based Computation Offloading Method in Vehicular Multiaccess Edge Computing Networks,'' \textit{IEEE Internet Things J.}, vol. 7, no. 6, pp. 4987--4996, 2020. 

\bibitem{22} G. Gao, M. Xiao, J. Wu, H. Huang, S. Wang, and G. Chen, ``Auction-based VM Allocation for Deadline-Sensitive Tasks in Distributed Edge Cloud,'' \textit{IEEE Trans. Services Comput}., pp. 1--1, 2019. 

\bibitem{23} M. Liwang, S. Dai, Z. Gao, Y. Tang, and H. Dai, ``A Truthful Reverse-Auction Mechanism for Computation Offloading in Cloud-Enabled Vehicular Network,'' \textit{IEEE Internet of Things Journal}, vol. 6, no. 3, pp. 4214--4227, 2019. 

\bibitem{24} Z. Gao, M. LiWang, S. Hosseinalipour, H. Dai, and X. Wang, ``A Truthful Auction for Graph Job Allocation in Vehicular Cloud-Assisted Networks,'' arXiv preprint arXiv:2003.12631, 2020. 

\bibitem{25} S. E. Khatib, and F. D. Galinan, ``Negotiating Bilateral Contracts in Electricity Markets,'' \textit{IEEE Trans. Power Syst., vol}. 22, no. 2, pp. 553--562, 2007. 

\bibitem{26} A. J. Conejo, R. Garcia-Bertrand, M. Carrion, \'{A}. Caballero, and A. de Andr\'{E}s, ``Optimal Involvement in Futures Markets of a Power Producer,'' \textit{IEEE Trans. Power Syst}., vol. 23, no. 2, pp. 703--711, 2008. 

\bibitem{27} J. M. Morales, S. Pineda, A. J. Conejo, and M. Carrion, ``Scenario Reduction for Futures Market Trading in Electricity Markets,'' \textit{IEEE Trans. Power Syst}., vol. 24, no. 2, pp. 878--888, 2009. 

\bibitem{28} X. Wang, X. Chen, and W. Wu, ``Optimal Dynamic Hedging of Electricity Futures based on Copula-GARCH Models,'' \textit{IEEE Int. Conf. Ind. Eng. \& Eng. Manag}., Macao, China, Dec. 2010, pp. 1--5. 

\bibitem{29} H. Algarvio, F. Lopes, J. A. M. Sousa, and J. Lagarto, ``Power Producers Trading Electricity in Both Pool and Forward Markets,'' \textit{IEEE Int. Workshop Database Expert Syst. Appl}., Munich, Germany, Sep. 2014, pp. 1--5. 

\bibitem{30} S. Mosquera-L\'{o}pez, and A. Nursimulu, ``Drivers of Electricity Price Dynamics: Comparative Analysis of Spot and Futures Markets,'' \textit{Energy Policy}, vol. 126, pp. 76--87, 2019. 

\bibitem{31} H. Li, T. Shu, F. He, and J. B. Song, ``Futures Market for Spectrum Trade in Wireless Communications: Modeling, Pricing and Hedging,'' \textit{IEEE Int. Conf. Global. Commun. (GLOBECOM)}, Atlanta, GA, USA, Dec. 2013, pp. 1--6. 

\bibitem{32} L. Gao, B. Shou, Y. J. Chen, and J. Huang, ``Combining Spot and Futures Markets: A Hybrid Market Approach to Dynamic Spectrum Access,'' \textit{Operations Res.}, vol. 64, no. 4, pp. 794--821, 2016. 

\bibitem{33} S. Sheng, R. Chen, P. Chen, X. Wang, and L. Wu, ``Futures-based Resource Trading and Fair Pricing in Real-Time IoT Networks,'' \textit{IEEE Wireless Commun. Lett}., vol. 9, no. 1, pp. 125--128, 2019. 

\bibitem{34} K. Vanmechelen, W. Depoorter, and J. Broeckhove, ``Combining Futures and Spot Markets: A Hybrid Market Approach to Economic Grid Resource Management,'' \textit{J. Grid Comput}. Vol. 9, no. 1, pp. 81--94, 2011. 

\bibitem{35} Y. Hsu, E. Modiano, and L. Duan, ``Scheduling Algorithms for Minimizing Age of Information in Wireless Broadcast Networks with Random Arrivals,'' \textit{IEEE Trans. Mobile Comput}., pp. 1--1, 2019. 

\bibitem{36} S. Ko, K. Han, and K. Huang, ``Wireless Networks for Mobile Edge Computing: Spatial Modeling and Latency Analysis,'' \textit{IEEE Trans. Wireless Commun}., vol. 17, no. 8, pp. 5225--5240, 2018. 

\bibitem{37} J. Wang, R. Chen, J. Huang, F. Shu, Z. Chen, and G. Min, ``Multiple-Antenna Spectrum Sensing Method with Random Arrivals of Primary Users,'' \textit{IEEE Trans. Veh. Technol.}, vol. 67, no. 9, pp. 8978--8983, 2018. 

\bibitem{38} D. Xu, Y. Sun, D. W. K. Ng, and R. Schober, "Multiuser MISO UAV Communications in Uncertain Environments with No-Fly Zones: Robust Trajectory and Resource Allocation Design," \textit{IEEE Trans. Commun}., vol. 68, no. 5, pp. 3153-3172, 2020.

\bibitem{39} L. Bai, R. Han, J. Liu, Q. Yu, J. Choi, and W. Zhang, ``Air-to-Ground Wireless Links for High-Speed UAVs,'' \textit{IEEE J. Sel. Areas Commun}., pp. 1--1, 2020. 

\bibitem{40} R. T. Marler, and J. S. Arora, ``The Weighted Sum Method for Multi-Objective Optimization: New Insights\textit{,'' Structural and Multidisciplinary Optimization}, vol. 41, no. 6, pp. 853--862, 2010. 

\bibitem{41} K. Deb, ``Multi-Objective Optimization Using Evolutionary Algorithms,'' \textit{John Wiley \& Sons}, 2001. 

\bibitem{42} A. Konak, D. W. Coit, and A. E. Smith, ``Multi-Objective Optimization using Genetic Algorithms: A Tutorial,'' \textit{Reliability Engineering \& System Safety}, vol. 91, no. 9, pp. 992--1007, 2006. 

\bibitem{43} R. M. Corless, G. H. Gonnet, D.E.G. Hare, D. J. Jeffrey, and D. E. Knuth, ``On the LambertW Function,'' \textit{Advances in Comput. mathematics}, vol. 5, no. 1, pp. 329--359, 1996. 

\bibitem{44} https://www.5gamericas.org/wp-content/uploads/2019/07/5G\_Americas\_ URLLLC\_White\_Paper\_Final\_\_updateJW.pdf 

\end{thebibliography}
\end{document}